\documentclass{lmcs}

\makeatletter
\@ifclassloaded{lms}{\newif\iflmsbuild\lmsbuildtrue}{\newif\iflmsbuild\lmsbuildfalse}
\makeatother

\usepackage{graphicx}%
\usepackage{multirow}%
\usepackage{amsmath,amssymb,amsfonts}%
\iflmsbuild
\else
\usepackage{amsthm}%
\fi
\usepackage{mathrsfs}%
\usepackage[overcommands]{overarrows}%
\usepackage{xcolor}
\usepackage{enumitem}
\usepackage{stmaryrd}
\usepackage{cmll} 
\usepackage{tikz-cd}
\usepackage{booktabs}
\usepackage[colorlinks=true,linkcolor=blue,citecolor=blue,urlcolor=blue]{hyperref}
\usepackage{cleveref}

\iflmsbuild

\fi

\iflmsbuild
\else
\makeatletter
\def\@setabstracta{%
  \ifvoid\abstractbox
  \else
    \skip@20\p@ \advance\skip@-\lastskip
    \advance\skip@-\baselineskip \vskip\skip@
    \unvbox\abstractbox
    \prevdepth\z@
  \fi
}
\makeatother
\fi

\iflmsbuild
\newtheorem{theorem}{Theorem}
\newtheorem{proposition}[theorem]{Proposition}
\newtheorem{lemma}[theorem]{Lemma}
\newtheorem{corollary}[theorem]{Corollary}

\newnumbered{definition}[theorem]{Definition}
\newnumbered{example}[theorem]{Example}
\newnumbered{remark}[theorem]{Remark}
\newnumbered{assumption}{Assumption}
\else
\theoremstyle{plain}
\newtheorem{theorem}{Theorem}
\newtheorem{proposition}[theorem]{Proposition}
\newtheorem{lemma}[theorem]{Lemma}
\newtheorem{corollary}[theorem]{Corollary}

\theoremstyle{definition}
\newtheorem{definition}[theorem]{Definition}
\newtheorem{example}[theorem]{Example}

\theoremstyle{remark}
\newtheorem{remark}[theorem]{Remark}
\newtheorem*{assumption*}{Assumption}
\fi

\newcommand{\R}{\mathbb{R}}
\newcommand{\op}{\mathrm{op}}
\newcommand{\Rbar}{\overline{\mathbb{R}}}

\newcommand{\cat}[1]{\mathcal{#1}}
\DeclareMathOperator{\Nuc}{\mathrm{Nuc}}
\DeclareMathOperator{\pnuc}{\mathbb{P}\mathrm{Nuc}}

\DeclareMathOperator{\fix}{\mathrm{Fix}}

\newcommand{\Ex}{\mathrm{Ex}}

\newcommand{\Mtc}{N}

\newcommand{\clL}[1]{\operatorname{cl}_L(#1)}
\newcommand{\clR}[1]{\operatorname{cl}_R(#1)}

\newcommand{\paren}[1]{\left(#1\right)}

\newcommand{\set}[1]{\left\{#1\right\}}


\begin{document}

\title{A Calculus of Types in Isbell Nuclei}


\author[J.-L.~Gastaldi]{Juan-Luis Gastaldi\lmcsorcid{0000-0003-0494-5266}}[a]
\author[S.~Jarvis]{Samantha Jarvis\lmcsorcid{0009-0004-5474-3739}}[b]
\author[T.~Seiller]{Thomas Seiller\lmcsorcid{0000-0001-6313-0898}}[c]
\author[J.~Terilla]{John Terilla\lmcsorcid{0000-0001-5456-3544}}[b]

\address{ETH Zurich, Zurich, Switzerland}	
\email{juan.luis.gastaldi@inf.ethz.ch}  

\address{City University New York, New York, USA}	
\email{samantha.jarvis@qc.cuny.edu, jterilla@gc.cuny.edu}  

\address{CNRS, Paris, France}	
\email{thomas.seiller@cnrs.fr}  





\begin{abstract}
  We identify a natural type calculus arising from a set with an associative
  product and a real-valued measurement.  The additivity defect of the
  measurement defines weighted orthogonality between weighted elements; the same
  function is also an \(\Rbar\)-valued profunctor.  The weighted orthogonality
  relation generates closed left and right types, and the paired closed types
  are identified with the points of the Isbell nucleus of that profunctor.  This
  identification gives coordinates in which the product can be studied as part
  of the type calculus.  Binary closed
  products need not be associative.  Keeping both contexts
  leads to a ternary measurement.  The resulting middle types carry an
  associative product, one-sided unit structure, and two oriented residuals,
  giving a noncommutative Lambek-like calculus.  In
  finite real coordinates, enriched Isbell nuclei have projective geometric
  structure, and the identification above places the type calculus in that
  geometric background.
\end{abstract}

\maketitle

\section{Introduction}\label{sec:introduction}

This paper joins two constructions with different mathematical origins. The first is the proofs-as-programs tradition of linear logic, where logical structure is read from composition and execution. The second is the enriched Isbell theory of profunctors, where a quantitative relation induces an adjunction between presheaves and copresheaves whose fixed points form a distinguished category called a nucleus. The types of the first construction and the fixed points of the second arise by different methods in different settings. This paper shows that they are the same objects, and develops the consequences of that identification.

In the first perspective, a linear realisability situation consists of a set \(\cat C\), an
associative execution product \(\cat C\times \cat C\to \cat C\), written \(ab\),
and a measurement \(p\colon \cat C\to \R\).  No compatibility between execution
and measurement is assumed.  Instead one forms a new function
\begin{equation}\label{eq:binary-measurement}
  M(a,b)=p(ab)-p(a)-p(b).
\end{equation}
and declares weighted elements \((a,\alpha),(b,\beta)\in \cat C\times \R\) to be orthogonal when
\[
  \alpha+\beta\le M(a,b).
\]
Left and right types are the biorthogonally closed subsets for this relation.  The types are not stipulated in advance; they are generated from execution and measurement.

This is an abstraction of a familiar theme in linear logic and realisability. In Girard's linear logic and geometry of interaction, and in the classical realisability tradition, logical formulas are controlled by orthogonality, duality, and execution. Seiller's linear realisability framework provides a general setting for this pattern. Here we work in a particularly lean real-valued case: one keeps the execution product and a numerical measurement, and then orthogonality generates the types. The simplicity of the datum $(\cat C,\Ex,p)$ leaves room for noncommutative execution, for asymmetric linear implication, and for an oriented type theory adapted to settings where stricter symmetric forms of linear logic may not fit.

By viewing the set \(\cat C\) as a discrete
\(\Rbar\)-category, the same formula \(M(a,b)=p(ab)-p(a)-p(b)\) defines an
\(\Rbar\)-enriched profunctor.  Its Isbell conjugates define adjoint functors between enriched pre and copresheaves:  if $f,g:\cat C \to \Rbar$, they are given by
\[
  M^*f(b)=\inf_{a\in\cat C}\bigl(M(a,b)-f(a)\bigr),
  \qquad
  M_*g(a)=\inf_{b\in\cat C}\bigl(M(a,b)-g(b)\bigr).
\]
The nucleus \(\Nuc(M)\) is the fixed-point locus of this adjunction. It is a category whose objects are the pairs
\((f,g)\) with $M^*f=g$ and $M_*g=f$.  Our first main result identifies the two constructions: the biorthogonally closed types of the linear realisability situation are precisely the nuclear fixed points of $M$.
The assignments
\[
  A\longmapsto \sup\{\alpha\in\R\mid (a,\alpha)\in A\},
  \qquad
  f\longmapsto \{(a,\alpha)\mid \alpha\le f(a)\},
\]
implement the correspondence.

Because the two sides contribute very different structures, this alignment has many consequences. From the Isbell side, the finite real case inherits the projective metric and polyhedral geometry of the companion paper \cite{GastaldiJarvisSeillerTerillaNucleus}: witness cells, event loci, order chambers, threshold lattice towers, and the pointed gap matrix, whose entries record exactly how the metric and polyhedral structures interact. From the realisability side, the execution product contributes the oriented logical structure of insertion.  A product placed inside a larger execution has a left boundary and a right boundary, and each orientation determines its own residual.  From this one obtains middle types, their execution product, and a corresponding noncommutative Lambek calculus.

This oriented logical structure appears when one asks how types compose.
Execution on \(\cat C\) extends to weighted elements of \(\cat C\times\R\),
and hence to subsets of \(\cat C\times\R\).  After biorthogonal closure, the
induced product on types need not be associative; a small finite example
makes this failure explicit.  Conceptually, one-sided binary closure does
not retain both boundary contexts of a composite.  The types introduced above
are governed by the two-variable measurement \eqref{eq:binary-measurement},
whereas the testing of a composite requires both a left context and a right
context.  This leads to the ternary measurement
\begin{equation}\label{eq:ternary-measurement}
  M_3(x,b,z)=p(xbz)-p(x)-p(b)-p(z)
\end{equation}
and hence to middle types defined by a middle-peripheral biorthogonality relation on $(\cat C \times \R)\times (\cat C \times \R)^2$ where
\((b,\beta)\) is orthogonal to \(((a,\alpha),(c,\gamma))\) when
\[
  \alpha+\beta+\gamma\le M_3(a,b,c).
\]
On middle types, the biorthogonal closure of the setwise execution product defines an associative product \(\odot\).  When the execution monoid has a unit,
\(\odot\) has a unit. It also has two residuals, \(\multimap_l\) and
\(\multimap_r\), satisfying
\[
  A\odot B\subseteq C
  \quad\Longleftrightarrow\quad
  B\subseteq A\multimap_l C
  \quad\Longleftrightarrow\quad
  A\subseteq B\multimap_r C.
\]
Together with the composition and currying laws, these operations form the
noncommutative Lambek calculus carried by the linear realisability situation:
its rules are derived directly from weighted orthogonality and execution.

Categorically, the middle types are the points of the nucleus of the profunctor
\[
  \Mtc\colon \cat C\nrightarrow \cat C\times\cat C,
  \qquad
  \Mtc\bigl(b,(x,z)\bigr)=M_3(x,b,z).
\]
So \(\odot\) defines an associative product on \(\Nuc(\Mtc)\).  In the categorical picture, a middle type naturally creates its own measurement.  A point \((f,g)\in\Nuc(\Mtc)\) has
copresheaf coordinate
\[
  g(x,z)=\inf_{b\in\cat C}\bigl(M_3(x,b,z)-f(b)\bigr).
\]
For fixed middle type \(f\), this function \(g\colon \cat C\times\cat C\to\Rbar\) may be read as a binary profunctor \(\cat C\nrightarrow\cat C\).  Hence a middle type determines a new binary Isbell adjunction and a new nucleus \(\Nuc(g)\).  These are the derived nuclei of the middle type.

The derived-nuclei theorem shows that, in the middle arrangement,
the new binary nucleus produced by a middle type is controlled by the residuals
of the Lambek calculus:
\[
  v=f\multimap_l M^*u,
  \qquad
  u=f\multimap_r M_*v.
\]
Moreover \(u\) and \(v\) are fixed by the binary Isbell closure operators for
\(M\).  Thus the presheaf and copresheaf coordinates derived from the
measurement \(g\) are genuine left and right type coordinates for the original
binary realisability situation, though the pair \((u,v)\in \Nuc(g)\) need not be a point of \(\Nuc(M)\).

The same ternary measurement \(M_3\) has three one-variable arrangements,
according to whether one singles out the left, middle, or right coordinate.  The left and right arrangements of \(M_3\) give companion derived constructions, with the same three kinds of coordinates appearing in different roles and similar derived-nuclei theorems relating the derived types to the Lambek calculus.  The
two-out-of-three theorem says that the three arrangements have the same pairwise
intersection: compatibility with any two of them forces compatibility with the
third.  The common locus is the balanced locus, where every two from \(u,f,v\) is in the derived nucleus of the third, and the third can be recovered as a corresponding residual of the other two.  This theorem specifies the coherence of three nuclei constructions which are related by different internal organizations of a single triple relation.

These theorems demonstrate the value of this identification of linear-realisability types
with enriched Isbell nuclei.  From the logical side, one has a rich noncommutative Lambek-like calculus owing to the interaction between orthogonality and execution.  Yet a type, as a subset of $\cat C \times \R$ does not itself suggest a new binary measurement and a derived type theory.  From the nuclear side, the many derived nuclei produced by factoring ternary (or higher) measurements produce a collection of derived nuclei whose internal structure and organization is hard to see, but is illuminated by the logic of types.

The identification also brings the middle product into the projective geometry
of the nucleus.  In the finite real case, the copresheaf coordinate of \(X\odot Y\) is a
convolution envelope governed by
\begin{equation}\label{eq:quartic-measurement}
  M_4(a,b,c,d)=p(abcd)-p(a)-p(b)-p(c)-p(d),
\end{equation}
and the corresponding product gap
decomposes into input witness gaps with extra outer sharpness terms.  Thus the product-gap formula gives a first bridge from the Lambek calculus of types to a geometry of execution, showing how execution-controlled composition is presented in the projective geometry of nuclei.

The main results are organized as follows.

\paragraph{Types as binary nuclear fixed points
  (Theorem~\ref{thm:nucleus-types}).}
Left and right types can be identified with the fixed points of the Isbell closure operators for the binary
measurement \(M\).  Points of \(\Nuc(M)\) are precisely the paired types.

\paragraph{Peripheral products and the passage to middle types
  (Example~\ref{ex:left-product-nonassociativity}; Proposition~\ref{thm.midassoc}).}
Execution extends to paired types by taking setwise products and then
biorthogonally closing.  This produces the peripheral products.  An explicit
example shows that they need not be associative.  This leads to the introduction
of middle types.  On middle types, middle closure of the setwise execution product
gives an associative product \(\odot\).

\paragraph{The Lambek calculus of middle types
  (Propositions~\ref{thm.dotmultimap}, \ref{prop.adj}, and
  \ref{thm.partialtransitivity}).}
Under a unital hypothesis, the middle product has a unit, and it carries left and
right residuals \(\multimap_l,\multimap_r\).  The residuation, currying,
composition, and partial transitivity laws give the noncommutative Lambek
calculus naturally associated with execution.  The rules are derived directly from weighted orthogonality and
execution.

\paragraph{Internalization of derived nuclei
  (Theorem~\ref{thm:derived-formula}).}
Middle types are identified with the points of \(\Nuc(\Mtc)\), where \(\Mtc\) is a profunctor defined by the ternary measurement \(M_3\).
For a middle nuclear point \((f,g)\in\Nuc(\Mtc)\), the copresheaf coordinate
\(g\) defines a binary measurement \(\cat C\nrightarrow\cat C\) and hence has its own nucleus.  Every point
\((u,v)\in\Nuc(g)\) satisfies
\[
  v=f\multimap_l M^*u,
  \qquad
  u=f\multimap_r M_*v,
\]
and the same theorem gives explicit duals placing \(u\) in
\(\operatorname{im}M_*\) and \(v\) in \(\operatorname{im}M^*\).  Derived nuclei therefore return to the binary fixed-coordinate loci of \(M\),
and their coordinates are computed by the residuals of the Lambek calculus.

\paragraph{Two-out-of-three coherence
  (Theorem~\ref{thm:triple-coincidence}).}
The left, middle, and right arrangements of \(M_3\) give three fixed-point
conditions on triples \((u,f,v)\).  Their pairwise intersections agree with
their triple intersection.  Equivalently, compatibility with any two
arrangements forces compatibility with the third, and the common locus is the
balanced locus where each coordinate is recovered as the residual of the other
two.

\paragraph{Product envelopes and factor gaps
  (Corollary~\ref{cor:projective-middle-product} and
  Proposition~\ref{prop:convolution-gap-factorization}).}
In the finite real-valued situation, the middle product
descends to a well-defined operation on \(\pnuc(\Mtc)\), and its copresheaf is the
convolution envelope obtained from \(M_4\).  The factor-gap identity decomposes a
candidate product witness into shifted input gap terms together with the sharpness
conditions for the two outer minimizations; in particular, an
envelope-computing factorisation forces all four quantities to vanish.

\paragraph{Examples.}
Throughout the paper we use explicit finite examples to show that the main
results are sharp.  A single four-element execution monoid is small enough to be
computed directly and illustrates a number of phenomena: the left closed product
need not be associative; the two-out-of-three hypothesis in the coherence theorem
cannot be weakened to a single arrangement; and the left, middle, and right units
in the Lambek calculus are distinct, with the middle unit profile strictly below
the pointwise minimum of the one-sided unit profiles.

\subsection*{Relation to prior work}

The linear realisability framework used here belongs to the lineage of Girard's
linear logic and geometry of interaction~\cite{girard-linearlogic,girard-goi5},
ludics~\cite{girard-ludics}, and the classical realisability tradition of
Kleene, Kreisel, Krivine, and van Oosten~\cite{kleene,kreisel,krivine,vanoosten}.
Seiller introduced linear realisability~\cite{seiller-hdr} to isolate the common
structure of geometry-of-interaction models and their extensions
\cite{seiller-goim,seiller-goiadd}.  The Lambek calculus originates in Lambek's
work on the mathematics of sentence structure~\cite{Lambek1958}; in the present
paper its noncommutative residuated structure is derived from weighted execution
and middle orthogonality.

The Isbell nucleus originates in Isbell's work on adequate subcategories
\cite{Isbell1960Adequate}.  Modern treatments of Isbell duality and enriched
nuclei include Avery--Leinster~\cite{averyLeinster2021isbell} and Willerton's
work in the metric setting~\cite{willerton2013tight,willerton2014galois,willerton2015legendre},
with the enriched metric viewpoint going back to Lawvere~\cite{lawvere1973metric}
and Kelly~\cite{Kelly1982}.  Jarvis studies monoidal structures on nuclei of
profunctors in a compatible monoidal setting~\cite{Jarvis2025NucleusProfunctor}.

The geometric results used here are developed in the companion paper
\cite{GastaldiJarvisSeillerTerillaNucleus}.  There the projective nucleus of a
finite real profunctor is studied through its gap matrix, witness cells, event
loci, order chambers, formal concept lattice towers, and Chebyshev centers.  The
enrichment viewpoint also appears in work on tropical convexity, directed
metrics, and language-like structures, including
\cite{develinSturmfels2004tropical,elliott2017fuzzy,fujii2019enriched,%
  GaubertVlassopoulos2024DirectedMetric,BradleyTerillaVlassopoulos2022EnrichedLanguage,%
  BradleyGastaldiTerilla2024,BradleyVigneaux2025MagnitudeTexts}.

\medskip\noindent
\textbf{Organization.}
Section~\ref{sec:nucleus} establishes the \(\Rbar\)-enriched conventions and recalls
Isbell duality for profunctors.  Section~\ref{sec:types} introduces linear
realisability situations, weighted orthogonality, and types.
Section~\ref{sec:correspondence} proves the nucleus--types correspondence.
Section~\ref{sec:execution-types} explains how execution extends to peripheral
types and analyzes the failure of associativity there.  Section~\ref{sec:middle}
introduces middle types and proves associativity of the middle product.
Section~\ref{sec:categorical-structure} develops units, residuals, and the
Lambek-style calculus of linear arrows.  Section~\ref{sec:derived} proves the
internalization theorem for derived nuclei, the corrected convolution formulas
for middle products, projective well-definedness of the product, and the ternary
coherence theorem.  Section~\ref{sec:product-geometry} treats the finite real
case and records the product envelope and factor-gap identities connecting
\(\odot\) with the witness geometry of the middle nucleus.

\section{Isbell duality over the extended reals}\label{sec:nucleus}

This section fixes the enriched conventions needed later to compare two kinds
of closure.  On the logical side, types will be closed subsets for a weighted
orthogonality relation.  On the Isbell side, the same objects will appear as
fixed points of closure operators associated to a profunctor.  These conventions
provide the arithmetic, profunctor, and Isbell-nucleus language for the binary
correspondence and for the derived nuclei of Section~\ref{sec:derived}.

\subsection{The arithmetic of $\Rbar$}\label{subsec:Rbar-cats}
Let $\Rbar=[-\infty,+\infty]$ regarded as a poset category where a morphism $x \to y$ exists precisely when $x\le y$, and arbitrary limits and colimits are
infima and suprema. The monoidal structure used throughout this paper is an extension of addition on $\R$ to $\Rbar$.  Its unit is $0$, and the endpoint convention is
\[
  -\infty+y=-\infty
  \qquad
  \text{for all }y\in\Rbar,
\]
including $y=+\infty$.  This is not an incidental convention.  We use $\Rbar$
as a closed ordered monoidal base category: for every $y\in\Rbar$, translation
$x\mapsto x+y$ is required to preserve arbitrary suprema and hence to have a
right adjoint.  In particular it preserves the bottom element, forcing
$-\infty+y=-\infty$. We write this right adjoint as $[y,-]$, or more suggestively as subtraction \(-y\):
\[
  x+y\le z
  \quad\Longleftrightarrow\quad
  x\le [y,z]=z-y.
\]
Equivalently,
\begin{equation}\label{eq:minus}
  z-y=\sup\{x\in\Rbar\mid x+y\le z\}.
\end{equation}
For finite \(y,z\), this is ordinary subtraction.  At the endpoints,
the notation \(z-y\) always means the residual defined by
\eqref{eq:minus}.  Thus subtraction by an endpoint should not be read
as addition by an opposite endpoint: for instance,
\[
  -\infty-(-\infty)=+\infty,
  \qquad\text{but}\qquad
  -\infty+(+\infty)=-\infty .
\]

The use of both endpoints is essential for the type story.  A subset
$A\subseteq C\times\R$ determines a function
\[
  c\longmapsto \sup\{\alpha\in\R\mid (c,\alpha)\in A\}.
\]
The value is $-\infty$ over an empty fibre and may be $+\infty$ over a fibre
containing arbitrarily large real weights.  Thus the Isbell coordinates of
types naturally live in $\Rbar$, even when the measurement function in the
realisability situation is real-valued.

\subsection{Categories, presheaves, and profunctors}

\begin{definition}\label{def:Rbar-category}
  A small $\Rbar$-\emph{category} $\cat C$ consists of a set
  $\mathrm{Ob}(\cat C)$ and hom-values
  $\cat C(c,c')\in\Rbar$ satisfying, for all $c,c',c''$,
  \begin{subequations}\label{eq:Rbar-axioms}
    \begin{align}
      0                           & \le \cat C(c,c),   \\
      \cat C(c,c')+\cat C(c',c'') & \le \cat C(c,c'').
    \end{align}
  \end{subequations}
\end{definition}
Note that it is not required that $\cat C(c,c)=0$, though that will be the case for the categories encountered in this paper.

The base $\Rbar$ is itself an $\Rbar$-category, with hom-value
\[
  \Rbar(x,y)=y-x,
\]
where the right-hand side is the residual \eqref{eq:minus}.  The opposite
$\cat C^{\op}$ has the same objects as $\cat C$ and hom-values
$\cat C^{\op}(c,c')=\cat C(c',c)$.

\begin{definition}\label{def:Rbar-functor}
  An $\Rbar$-\emph{functor} $F:\cat C\to\cat D$ is a function on objects such
  that
  \[
    \cat C(c,c')\le \cat D(Fc,Fc')
  \]
  for all objects $c,c'$ of $\cat C$.
\end{definition}

For $\Rbar$-categories $\cat A$ and $\cat B$, we write
$[\cat A,\cat B]$ for the $\Rbar$-category of $\Rbar$-functors
$\cat A\to\cat B$, where hom-values are computed pointwise:  if $F,G:\cat A\to\cat B$ are two such functors, then
\begin{equation}\label{eq:functor-hom}
  [\cat A,\cat B](F,G)=\inf_{a\in\cat A}\cat B(Fa,Ga).
\end{equation}
When the ambient functor category is clear, we write $[F,G]$ for this
hom-value.

A \emph{presheaf} on $\cat C$ is an $\Rbar$-functor
$f:\cat C^{\op}\to\Rbar$.  A \emph{copresheaf} on $\cat D$ is an
$\Rbar$-functor $g:\cat D\to\Rbar$.  In the Isbell adjunction below,
copresheaves are regarded as objects of the opposite enriched category
$[\cat D,\Rbar]^{\op}$.  Thus the underlying functions are still
$g:\cat D\to\Rbar$, but the enriched order is reversed. For presheaves $f,f':\cat C^{\op}\to\Rbar$ and copresheaves
$g,g':\cat D\to\Rbar$, the hom-values are
\[
  \begin{aligned}
    \relax[f,f']                  & = \inf_{c\in\cat C}\left(f'(c)-f(c)\right), \\
    [g,g']_{[\cat D,\Rbar]^{\op}} & = [\cat D,\Rbar](g',g)
    = \inf_{d\in\cat D}\left(g(d)-g'(d)\right).
  \end{aligned}
\]
Consequently, the induced order on presheaves is the pointwise order:
$f\le f'$ iff $f(c)\le f'(c)$ for all $c$.  By contrast, the induced
order on copresheaves as objects of $[\cat D,\Rbar]^{\op}$ is the
opposite pointwise order:
\[
  g\le_{[\cat D,\Rbar]^{\op}} g'
  \quad\Longleftrightarrow\quad
  g'(d)\le g(d) \text{ for all } d.
\]
To keep the notation readable, however, unadorned inequalities
$f\le f'$ and $g\le g'$ will always mean pointwise inequalities unless
explicitly stated otherwise.  This agrees with the enriched order for
presheaves, but not for copresheaves viewed in $[\cat D,\Rbar]^{\op}$.

The enriched Yoneda lemma gives the formulas
\[
  \relax[\cat C(-,c),f]=f(c),
  \qquad
  [g,\cat D(d,-)]_{[\cat D,\Rbar]^{\op}}=g(d).
\]
These formulas fix the representable presheaves and copresheaves that occur
when a profunctor is evaluated in one variable. The tensor product of $\Rbar$-categories $\cat C$ and $\cat D$ is the
$\Rbar$-category $\cat C\otimes\cat D$ whose objects are pairs $(c,d)$ and
whose hom-values are
\[
  (\cat C\otimes\cat D)\left((c,d),(c',d')\right)
  =\cat C(c,c')+\cat D(d,d').
\]

\begin{definition}\label{def:profunctor}
  A \emph{profunctor} $M:\cat C\nrightarrow\cat D$ is an $\Rbar$-functor
  \[
    M:\cat C^{\op}\otimes\cat D\to\Rbar.
  \]
  Equivalently, it is a function
  $M:\mathrm{Ob}(\cat C)\times\mathrm{Ob}(\cat D)\to\Rbar$ such that, for all
  $c,c'\in\cat C$ and $d,d'\in\cat D$,
  \begin{equation}\label{eq:profunctor-axiom}
    \cat C(c',c)+\cat D(d,d') \le [M(c,d),M(c',d')]
    = M(c',d')-M(c,d).
  \end{equation}
\end{definition}

A set $S$ will often be viewed as a discrete $\Rbar$-category, again denoted
$S$, by
\[
  S(s,s')=
  \begin{cases}
    0       & s=s',     \\
    -\infty & s\neq s'.
  \end{cases}
\]
For discrete $\cat C$ and $\cat D$, the profunctor condition
\eqref{eq:profunctor-axiom} imposes no restriction: every function
$M:\cat C\times\cat D\to\Rbar$ is a profunctor.  This is the case when the
binary measurement
\[
  M(a,b)=p(ab)-p(a)-p(b)
\]
is regarded as an $\Rbar$-profunctor.

\subsection{Isbell duality and the nucleus}\label{subsec:nucleus}

Let $M:\cat C\nrightarrow\cat D$ be a profunctor.  For each $d\in\cat D$,
the column $M(-,d)$ is the $\cat C$-indexed slice of $M$, hence a
presheaf on $\cat C$; for each $c\in\cat C$, the row $M(c,-)$ is the
$\cat D$-indexed slice of $M$, hence a copresheaf on $\cat D$.  The
terms ``row'' and ``column'' are only mnemonic: no ordering or enumeration
of the objects of $\cat C$ or $\cat D$ is being chosen.  If $\cat C$ and
$\cat D$ were finite discrete categories with chosen enumerations, then
the values $M(c,d)$ could be displayed as a rectangular matrix, and these slices
would be its literal columns and rows.  In general they are simply the one-variable slices of the profunctor.

The Isbell conjugates extend these
assignments from rows and columns to arbitrary presheaves and copresheaves.

\begin{definition}\label{def:Isbell-conjugates}
  Define maps
  \[
    \begin{aligned}
      M^*\colon [\cat C^{\op},\Rbar] & \to [\cat D,\Rbar]^{\op}, \\
      M_*\colon [\cat D,\Rbar]^{\op} & \to [\cat C^{\op},\Rbar]
    \end{aligned}
  \]
  by
  \begin{align}
    (M^*f)(d)
     & := [f,M(-,d)]
    = \inf_{c\in\cat C}\left(M(c,d)-f(c)\right),
    \label{eq:Mupperstar}                   \\
    (M_*g)(c)
     & := [M(c,-),g]_{[\cat D,\Rbar]^{\op}}
    = \inf_{d\in\cat D}\left(M(c,d)-g(d)\right).
    \label{eq:Mlowerstar}
  \end{align}
  The maps $M^*$ and $M_*$ are the Isbell conjugates induced by $M$.
\end{definition}

As usual, minus signs in \eqref{eq:Mupperstar}--\eqref{eq:Mlowerstar} are residuals in $\Rbar$, see Equation \eqref{eq:minus}.

\begin{proposition}\label{prop:Isbell-adjunction}
  The assignments \eqref{eq:Mupperstar}--\eqref{eq:Mlowerstar} define adjoint
  $\Rbar$-functors, meaning that
  \[
    [\cat D,\Rbar]^{\op}(M^*f,g)=[\cat C^{\op},\Rbar](f,M_*g)
  \]
  for all $f\in[\cat C^{\op},\Rbar]$ and
  $g\in[\cat D,\Rbar]^{\op}$.  In symbols, \(M^*\dashv M_*\).
\end{proposition}

\begin{proof}
  We record the verification in the present $\Rbar$-valued notation.  The
  functoriality calculation is the usual one for Isbell conjugacy.  If
  $\delta=[f,f']=\inf_c(f'(c)-f(c))$, then
  $f(c)+\delta\le f'(c)$ for every $c$.  By the residuation law, this implies
  \[
    \delta+\left(M(c,d)-f'(c)\right)\le M(c,d)-f(c).
  \]
  Since $(M^*f')(d)\le M(c,d)-f'(c)$ for every $c$, we obtain
  \[
    \delta+(M^*f')(d)\le M(c,d)-f(c)
  \]
  for every $c$, and hence
  $\delta+(M^*f')(d)\le (M^*f)(d)$.  This is precisely the enriched
  functoriality inequality for $M^*$.  The proof for $M_*$ is the same. For the adjunction, residuation by $g(d)$ preserves infima, since it is a
  right adjoint.  Hence
  \[
    \begin{aligned}
      \relax[\cat D,\Rbar]^{\op}(M^*f,g)
       & =
      \inf_d\left(\inf_c\left(M(c,d)-f(c)\right)-g(d)\right) \\
       & =
      \inf_d\inf_c\left(M(c,d)-f(c)-g(d)\right)              \\
       & =
      \inf_c\left(\inf_d\left(M(c,d)-g(d)\right)-f(c)\right) \\
       & =
      \relax[\cat C^{\op},\Rbar](f,M_*g).
    \end{aligned}
  \]
\end{proof}

\begin{lemma}\label{lem:antitone}
  With respect to pointwise order, the maps $M^*$ and $M_*$ are order-reversing:
  \[
    f\le f'\Rightarrow M^*f'\le M^*f,
    \qquad
    g\le g'\Rightarrow M_*g'\le M_*g.
  \]
\end{lemma}

\begin{proof}
  If $f\le f'$, then
  $M(c,d)-f'(c)\le M(c,d)-f(c)$ for all $c,d$, because the residual is
  antitone in its second variable.  Taking infima gives
  $M^*f'\le M^*f$.  The proof for $M_*$ is identical.
\end{proof}

Thus the composites
\[
  \begin{aligned}
    \mathrm{cl}_{\cat C}:=M_*M^*\colon [\cat C^{\op},\Rbar] & \to[\cat C^{\op},\Rbar], \\
    \mathrm{cl}_{\cat D}:=M^*M_*\colon [\cat D,\Rbar]^{\op} & \to[\cat D,\Rbar]^{\op}
  \end{aligned}
\]
are monotone for pointwise order.  The adjunction $M^*\dashv M_*$ implies that
they are closure operators:
\begin{equation}\label{eq:closure}
  \begin{aligned}
    f & \le \mathrm{cl}_{\cat C}(f), & \mathrm{cl}_{\cat C}^2 & =\mathrm{cl}_{\cat C}, \\
    g & \le \mathrm{cl}_{\cat D}(g), & \mathrm{cl}_{\cat D}^2 & =\mathrm{cl}_{\cat D}.
  \end{aligned}
\end{equation}

\begin{definition}\label{def:nucleus}
  The \emph{nucleus} of $M$ is the $\Rbar$-category $\Nuc(M)$ whose objects are
  pairs
  \[
    \Nuc(M)=\set{(f,g)\mid f\in[\cat C^{\op},\Rbar],\ g\in[\cat D,\Rbar]^{\op},\ g=M^*f,\ f=M_*g},
  \]
  with hom-values inherited from either side:
  \[
    \Nuc(M)\paren{(f,g),(f',g')}=[\cat C^{\op},\Rbar](f,f')
    =[\cat D,\Rbar]^{\op}(g,g').
  \]
  The equality follows from the adjunction $M^*\dashv M_*$.
\end{definition}

The fixed-point description of the nucleus is used throughout this paper.

\begin{proposition}\label{prop:nucleus-fixedpoints}
  There are canonical isomorphisms of $\Rbar$-categories
  \[
    \Nuc(M)\cong\fix(\mathrm{cl}_{\cat C})
    \cong\fix(\mathrm{cl}_{\cat D})
    \cong\mathrm{im}(M^*)
    \cong\mathrm{im}(M_*),
  \]
  where $\fix(\mathrm{cl}_{\cat C})=\set{f\mid \mathrm{cl}_{\cat C}(f)=f}$ and
  similarly for $\cat D$, and $\mathrm{im}(M^*)$ and $\mathrm{im}(M_*)$ denote
  the full subcategories spanned by objects of the corresponding forms.
\end{proposition}

\begin{proof}
  The projection $(f,g)\mapsto f$ identifies $\Nuc(M)$ with
  $\fix(M_*M^*)$: indeed, $(f,g)\in\Nuc(M)$ if and only if
  $f=M_*g=M_*M^*f$.  The projection $(f,g)\mapsto g$ gives the dual
  identification with $\fix(M^*M_*)$. For any closure operator, fixed points and image agree.  Finally,
  $M_*$ already lands in $\fix(M_*M^*)$, because the expansion
  $g\le M^*M_*g$ and antitonicity of $M_*$ give
  $M_*M^*M_*g\le M_*g$, while expansion gives the reverse inequality.  Thus
  $M_*M^*M_*=M_*$.  Dually, $M^*M_*M^*=M^*$.
\end{proof}

\begin{remark}\label{rem:constructing-nucleus-points}
  Proposition~\ref{prop:nucleus-fixedpoints} gives an explicit way to produce
  objects of $\Nuc(M)$.  For any presheaf $f$, the pair
  $\left(\mathrm{cl}_{\cat C}(f),M^*f\right)$ lies in $\Nuc(M)$.  For any
  copresheaf $g$, the pair
  $\left(M_*g,\mathrm{cl}_{\cat D}(g)\right)$ lies in $\Nuc(M)$.
\end{remark}

\begin{remark}\label{cor:fiber-max}
  A presheaf $f$ is $\mathrm{cl}_{\cat C}$-closed if and only if it is the
  largest presheaf, for the pointwise order, among those with the same
  $M^*$-image:
  \[
    M^*h=M^*f \Longrightarrow h\le f.
  \]
  Dually, a copresheaf $g$ is $\mathrm{cl}_{\cat D}$-closed if and only if it is
  the largest copresheaf, for the pointwise order, among those with the same
  $M_*$-image.
\end{remark}

In the next section, we introduce a real linear realisability situation: a set \(\cat C\) equipped with an associative execution operation and a real-valued measurement $p$.  From these data
one obtains a binary measurement \(M(a,b)=p(ab)-p(a)-p(b)\); after viewing the
set \(\cat C\) as a discrete \(\Rbar\)-category, this measurement becomes an
\(\Rbar\)-profunctor of the kind studied above.  The Isbell-nuclear formalism
then provides one side of the comparison with the logical types generated by
orthogonality.

\section{Linear realisability and types}\label{sec:types}

Section~\ref{sec:nucleus} developed the profunctorial construction: a
profunctor \(M\) determines conjugate Isbell closure operators and a nucleus.
We now give the corresponding realisability construction.  A linear
realisability situation starts from a set \(\cat C\) equipped with an associative product and a real-valued
function; from these data we obtain a binary
measurement \(M\), a weighted orthogonality relation on \(\cat C\times\R\),
and the left and right types generated by the associated
orthogonal-complement operations.

\subsection{The realisability datum and its measurement}

\begin{definition}\label{def:linear-realisability-situation}
  A \emph{real linear realisability situation} is a triple
  \((\cat C,\Ex,p)\), where \(\cat C\) is a set, \(\Ex\colon
  \cat C\times\cat C\to\cat C\) is an associative execution product, and
  \(p\colon \cat C\to\R\) is a function.  We write
  \[
    \Ex(a,b)=ab.
  \]
\end{definition}

No unit is assumed, and no compatibility between \(p\) and execution is
assumed.  Instead we measure the defect of execution to be additive by the function
\(
M\colon \cat C\times\cat C\longrightarrow\R\) defined by
\begin{equation}\label{eq:defofM}
  \qquad
  M(a,b):=p(ab)-p(a)-p(b).
\end{equation}
We call $M$ the \emph{measurement} of the linear realisability situation.

\begin{proposition}[Trefoil identity]
  The measurement \(M\) satisfies
  \begin{equation}\label{trefoil}
    M(ab,c)+M(a,b)=M(a,bc)+M(b,c)
  \end{equation}
  for all \(a,b,c\in\cat C\).
\end{proposition}

\begin{proof}
  Associativity gives \(p((ab)c)=p(a(bc))=p(abc)\).  Therefore
  \begin{align*}
    M(ab,c)+M(a,b)
     & =p(abc)-p(ab)-p(c)+p(ab)-p(a)-p(b) \\
     & =p(abc)-p(a)-p(b)-p(c)             \\
     & =p(abc)-p(a)-p(bc)+p(bc)-p(b)-p(c) \\
     & =M(a,bc)+M(b,c).
  \end{align*}
\end{proof}
\begin{remark}
  The construction above should be compared with Seiller's
  linear-realisability framework.  There the measurement is a primitive binary
  map
  \[
    [|\,\cdot,\cdot\,|]_m\colon P\times P\to\Theta
  \]
  with values in a commutative group \(\Theta\), and the trefoil identity is
  assumed as an axiom; see \cite[Chapter 10]{seiller-hdr}.  In the present
  paper, the target group is \(\R\), no commutativity of the execution product
  is assumed, and the measurement is not specified independently.  Instead it
  is induced from the unary function \(p\) by Equation \eqref{eq:defofM}.
  Thus the trefoil identity follows formally.

  The same fact can be recognized in the real bar cochain complex of the
  associative product on \(\cat C\).  Let
  \[
    C^n(\cat C;\R)=\{Q:\cat C^n\to\R\}
  \]
  be the real bar cochains, and write \(\delta\) for the bar differential.  In
  low degrees,
  \[
    (\delta q)(a,b)=q(a)+q(b)-q(ab)
  \]
  for \(q:\cat C\to\R\), and
  \[
    (\delta Q)(a,b,c)
    =
    Q(b,c)-Q(ab,c)+Q(a,bc)-Q(a,b)
  \]
  for \(Q:\cat C^2\to\R\).  With these signs,
  \[
    M=\delta(-p).
  \]
  Hence the cocycle equation \(\delta M=0\) follows from \(\delta^2=0\).
  Written out, this cocycle equation is exactly the trefoil identity.  We will
  not otherwise use the bar complex; in what follows, we work directly with
  the displayed trefoil identity.
\end{remark}

For the link with the enriched constructions, regard \(\cat C\) as the
discrete \(\Rbar\)-category on this set.  Then the real-valued function \(M\)
becomes an \(\Rbar\)-valued profunctor
\[
  M\colon \cat C\nrightarrow\cat C
\]
by the inclusion \(\R\subset\Rbar\).  This is the datum to which the Isbell
constructions of Section~\ref{sec:nucleus} apply.

\subsection{Weighted orthogonality}

The type construction is generated by a weighted form of orthogonality.  For
\((a,\alpha),(b,\beta)\in\cat C\times\R\), define
\begin{equation}\label{eq:weighted-orthogonality}
  (a,\alpha)\perp(b,\beta)
  \quad\Longleftrightarrow\quad
  \alpha+\beta\le M(a,b).
\end{equation}
The unweighted inequality \(0\le M(a,b)\) may be used as a sign test for
interaction between \(a\) and \(b\), but the formal objects of this paper are
the weighted subsets of \(\cat C\times\R\) closed under
\eqref{eq:weighted-orthogonality}. For later multiplicative constructions, it is useful to record the associated
weighted execution product
\[
  (a,\alpha)(b,\beta):=(ab,\alpha+\beta-M(a,b)).
\]
Its associativity is exactly the trefoil identity: the two second coordinates
of \(((a,\alpha)(b,\beta))(c,\gamma)\) and
\((a,\alpha)((b,\beta)(c,\gamma))\) are
\[
  \alpha+\beta+\gamma-M(a,b)-M(ab,c)
  \quad\text{and}\quad
  \alpha+\beta+\gamma-M(b,c)-M(a,bc),
\]
which agree by \eqref{trefoil}.  The complement calculus below uses only the
relation \(\perp\); later sections use the displayed product when discussing
products of types.

\subsection{Complements and types}

For subsets \(A,B\subseteq\cat C\times\R\), define the left and right
orthogonal complements by
\[
  {}^\perp B
  :=
  \{(a,\alpha)\mid
  \alpha+\beta\le M(a,b)
  \text{ for all }(b,\beta)\in B\},
\]
and
\[
  A^\perp
  :=
  \{(b,\beta)\mid
  \alpha+\beta\le M(a,b)
  \text{ for all }(a,\alpha)\in A\}.
\]
The two constructions are oriented.  Since \(M(a,b)\) need not equal
\(M(b,a)\), the left and right complements need not agree.

\begin{definition}\label{def:left-right-closures}
  For \(A\subseteq\cat C\times\R\), its \emph{left closure} is
  \[
    \clL{A}:={}^\perp(A^\perp).
  \]
  Its \emph{right closure} is
  \[
    \clR{A}:=({}^\perp A)^\perp.
  \]
\end{definition}

\begin{definition}\label{def:left-right-types}
  A subset \(A\subseteq\cat C\times\R\) is a \emph{left type} if
  \(A={}^\perp B\) for some \(B\subseteq\cat C\times\R\).  A subset
  \(B\subseteq\cat C\times\R\) is a \emph{right type} if
  \(B=A^\perp\) for some \(A\subseteq\cat C\times\R\).  A \emph{paired type}
  is a pair \((A,B)\) such that
  \[
    A={}^\perp B
    \qquad\text{and}\qquad
    B=A^\perp.
  \]
\end{definition}

Thus left and right types are not primitive data.  They are the closed objects
generated by weighted orthogonality. We now describe the properties of our orthogonal-complement operation.

\begin{proposition}[Orthogonal-complement calculus]
  \label{prop:acontainedinaperpperp}
  \label{prop:threecomplements}
  \label{prop:typesdoubleperp}
  Let all sets below be subsets of \(\cat C\times\R\).
  \begin{enumerate}[label=\textup{(\roman*)}]
    \item If \(A_1\subseteq A_2\), then
          \(A_2^\perp\subseteq A_1^\perp\).  If \(B_1\subseteq B_2\), then
          \({}^\perp B_2\subseteq{}^\perp B_1\).

    \item For every \(A\) and \(B\),
          \[
            A\subseteq{}^\perp(A^\perp),
            \qquad
            B\subseteq({}^\perp B)^\perp.
          \]

    \item Alternating three complements cancels to one complement:
          \[
            {}^\perp\bigl(({}^\perp B)^\perp\bigr)={}^\perp B,
            \qquad
            \bigl({}^\perp(A^\perp)\bigr)^\perp=A^\perp.
          \]

    \item A subset \(A\) is a left type if and only if \(\clL{A}=A\).
          A subset \(B\) is a right type if and only if \(\clR{B}=B\).
          In particular, \(\clL{A}\) is a left type and \(\clR{B}\) is a right type for
          arbitrary \(A\) and \(B\).

    \item The assignments
          \[
            A\longmapsto A^\perp,
            \qquad
            B\longmapsto{}^\perp B
          \]
          restrict to inverse bijections between left types and right types.  Under
          these bijections, a left type \(A\) and a right type \(B\) correspond
          precisely when \((A,B)\) is a paired type.
  \end{enumerate}
\end{proposition}

\begin{proof}
  The first assertion is immediate from the definitions: enlarging the set to
  be tested against imposes more inequalities. For the second assertion, if \((a,\alpha)\in A\) and
  \((b,\beta)\in A^\perp\), then \((a,\alpha)\perp(b,\beta)\) by definition of
  \(A^\perp\).  Hence \((a,\alpha)\in{}^\perp(A^\perp)\).  The proof of
  \(B\subseteq({}^\perp B)^\perp\) is the same with left and right interchanged.

  For the first identity in (iii), apply (ii) to the set \({}^\perp B\) to get
  \[
    {}^\perp B\subseteq{}^\perp\bigl(({}^\perp B)^\perp\bigr).
  \]
  On the other hand, (ii) gives \(B\subseteq({}^\perp B)^\perp\), and
  antitonicity of the left complement gives
  \[
    {}^\perp\bigl(({}^\perp B)^\perp\bigr)\subseteq{}^\perp B.
  \]
  The second identity in (iii) is analogous. If \(A\) is a left type, say \(A={}^\perp B\), then (iii) gives
  \[
    \clL{A}={}^\perp(A^\perp)
    ={}^\perp\bigl(({}^\perp B)^\perp\bigr)
    ={}^\perp B=A.
  \]
  Conversely, if \(\clL{A}=A\), then \(A={}^\perp(A^\perp)\), so \(A\) is a
  left type.  The right-type assertion is parallel.  The final assertion then
  follows directly: if \(A\) is a left type, then \(A={}^\perp(A^\perp)\), so
  \(A\) is paired with \(A^\perp\); if \(B\) is a right type, then
  \(B=({}^\perp B)^\perp\), so \(B\) is paired with \({}^\perp B\).
\end{proof}

We write
\[
  \mathrm{LType}(M),\qquad
  \mathrm{RType}(M),\qquad
  \mathrm{PType}(M)
\]
for the left types, right types, and paired types determined by \(M\).  The
preceding proposition gives canonical bijections among these three
presentations:
\[
  \begin{tikzcd}
    & {\mathrm{PType}(M)} \\
    \\
    {\mathrm{RType}(M)} && {\mathrm{LType}(M)}
    \arrow["{(A,B)\mapsto B}"', from=1-2, to=3-1]
    \arrow["{(A,B)\mapsto A}", from=1-2, to=3-3]
    \arrow["{B\mapsto{}^\perp B}", shift left, from=3-1, to=3-3]
    \arrow["{A\mapsto A^\perp}", shift left, from=3-3, to=3-1]
  \end{tikzcd}
\]


The enriched language developed in Section~\ref{sec:nucleus} now meets the
construction from weighted orthogonality developed here.  In the next section we
introduce the maps \(\Omega\) and \(\varphi\), which translate between weighted
subsets of \(\cat C\times\R\) and \(\Rbar\)-valued coordinates, and use them to
compare the left and right types above with fixed points in the Isbell nucleus
of \(M\).

\section{The nucleus-types correspondence}\label{sec:correspondence}

Section~\ref{sec:nucleus} associates to a profunctor \(M\) an Isbell
adjunction and its nucleus.  Section~\ref{sec:types} associates to the same
measurement, when it comes from a real linear realisability situation, a
weighted orthogonality relation on \(\cat C\times\R\), together with its left
and right complement operations.  We now compare these two closure
constructions.

Throughout this section \(\cat C\) is the set from Section~\ref{sec:types},
viewed as a discrete \(\Rbar\)-category, and
\[
  M(a,b)=p(ab)-p(a)-p(b)
\]
is regarded as an \(\Rbar\)-valued profunctor by the inclusion
\(\R\subset\Rbar\).  Since \(\cat C\) is discrete, both
presheaves and copresheaves have the same underlying data: functions
\(\cat C\to\Rbar\).  The distinction between them is therefore not in their
coordinates, but in the enriched order conventions described in
Section~\ref{sec:nucleus}.  In the arguments below, whenever an order
comparison is needed, we state it in the pointwise order explicitly.

\subsection[Profiles and recovered weighted subsets]
{Profiles and recovered weighted subsets}

For a function \(f:\cat C\to\Rbar\), define the weighted subset
\[
  \Omega_f:=\{(a,\alpha)\in\cat C\times\R\mid \alpha\le f(a)\}.
\]
Although \(f\) takes values in \(\Rbar\), the subset \(\Omega_f\) lies in
\(\cat C\times\R\).  The endpoint values determine the exceptional fibres:
\(f(a)=+\infty\) gives the whole fibre over \(a\), while \(f(a)=-\infty\)
gives the empty fibre.

Conversely, any weighted subset \(X\subseteq\cat C\times\R\) has a profile
\[
  \varphi_X(c):=\sup\{\xi\in\R\mid (c,\xi)\in X\},
\]
where the supremum is taken in \(\Rbar\).

These two constructions are inverse in one direction only:
\[
  \varphi_{\Omega_f}=f,
  \qquad
  X\subseteq\Omega_{\varphi_X}.
\]
Thus passing from a function to its weighted subset and back recovers the
function exactly.  Passing from a subset to its profile and back generally
enlarges the subset: fibrewise, it replaces the original fibre by the lower
ray determined by its supremum.  For example, over a single object, the subsets
\[
  \{0\},\qquad (-\infty,0),\qquad (-\infty,0]
\]
all have profile \(0\), but reconstruction from the profile gives
\((-\infty,0]\).  We formalize this terminology:

\begin{definition}
  A weighted subset \(X\subseteq\cat C\times\R\) is \emph{recovered from its
    profile} if
  \(
  X=\Omega_{\varphi_X}.
  \)
\end{definition}

Equivalently, \(X\) is recovered from its profile if each fibre is either
empty, a closed lower ray \((-\infty,r]\) with \(r\in\R\), or all of \(\R\).
In other words, the fibres are down-closed and contain their supremum whenever
that supremum is finite.

The assignments
\[
  f\longmapsto\Omega_f,
  \qquad
  X\longmapsto\varphi_X
\]
therefore identify functions \(\cat C\to\Rbar\) with weighted subsets of $\cat C \times \R$ recovered from their profiles.

For some situations it is helpful to distinguish handedness: when a weighted subset is used in the left variable, we write
\[
  \lambda_A:=\varphi_A,
\]
and when it is used in the right variable, we write
\[
  \rho_B:=\varphi_B.
\]
These are the same profile construction, but the notation keeps track of which
coordinate of the Isbell adjunction the profile occupies: \(\lambda_A\) is the
left, or presheaf, coordinate, while \(\rho_B\) is the right, or copresheaf,
coordinate.

\subsection{Orthogonality in profile coordinates}

Recall from Section~\ref{sec:nucleus} that the Isbell conjugates associated to
\(M\) are
\[
  (M^*f)(b)=\inf_{a\in\cat C}\bigl(M(a,b)-f(a)\bigr),
  \qquad
  (M_*g)(a)=\inf_{b\in\cat C}\bigl(M(a,b)-g(b)\bigr),
\]
where subtraction means the residual in \(\Rbar\).  The first formula sends a
left coordinate to a right coordinate; the second sends a right coordinate to a
left coordinate.

\begin{proposition}\label{prop:orthogonality-isbell-profiles}
  For all \(A,B\subseteq\cat C\times\R\),
  \[
    A^\perp=\Omega_{M^*\lambda_A},
    \qquad
    {}^\perp B=\Omega_{M_*\rho_B}.
  \]
\end{proposition}

\begin{proof}
  We prove the first identity.  Let \((b,\beta)\in\cat C\times\R\).  By
  definition of the right complement,
  \[
    (b,\beta)\in A^\perp
    \quad\Longleftrightarrow\quad
    \alpha+\beta\le M(a,b)
    \text{ for all }(a,\alpha)\in A.
  \]
  Fixing \(a\), this condition over the fibre of \(A\) is equivalent to
  \[
    \lambda_A(a)+\beta\le M(a,b),
  \]
  because translation by the real number \(\beta\) preserves suprema in
  \(\Rbar\).  This remains meaningful at the endpoints: an empty fibre gives
  \(\lambda_A(a)=-\infty\) and no constraint, while an unbounded fibre gives
  \(\lambda_A(a)=+\infty\), which in the present real-valued measurement
  setting gives no real \(\beta\) satisfying the inequality. By residuation, the displayed inequality is equivalent to
  \[
    \beta\le M(a,b)-\lambda_A(a).
  \]
  Requiring this for every \(a\) is equivalent to
  \[
    \beta\le \inf_{a\in\cat C}\bigl(M(a,b)-\lambda_A(a)\bigr)
    =(M^*\lambda_A)(b).
  \]
  This is precisely \((b,\beta)\in\Omega_{M^*\lambda_A}\).  The proof of
  \({}^\perp B=\Omega_{M_*\rho_B}\) is the same calculation with the two
  variables interchanged.
\end{proof}

\begin{corollary}\label{cor:types-profile-complete}
  Every left type and every right type is recovered from its profile.
  More precisely, if \(A\) is a left type and \(B\) is a right type, then
  \[
    A=\Omega_{\lambda_A},
    \qquad
    B=\Omega_{\rho_B}.
  \]
\end{corollary}

\begin{proof}
  A right type \(B\) has the form \(X^\perp\) for some \(X\), and
  Proposition~\ref{prop:orthogonality-isbell-profiles} writes it as
  \(\Omega_{M^*\lambda_X}\).  A left type \(A\) has the form
  \({}^\perp Y\) and is handled by the second identity in the same
  proposition.  The identities with \(\lambda_A\) and \(\rho_B\) then follow.
\end{proof}

\subsection{Fixed points and paired types}

The preceding proposition identifies orthogonal complement with Isbell
conjugacy.  Applying it twice identifies biorthogonal closure with the two
Isbell closure operators.

\begin{corollary}\label{cor:single-coordinate-fixedpoints}
  Let \(f:\cat C\to\Rbar\) be a left coordinate and
  \(g:\cat C\to\Rbar\) a right coordinate.  Then
  \[
    \Omega_f \text{ is a left type}
    \quad\Longleftrightarrow\quad
    M_*M^*f=f,
  \]
  and
  \[
    \Omega_g \text{ is a right type}
    \quad\Longleftrightarrow\quad
    M^*M_*g=g.
  \]
\end{corollary}

\begin{proof}
  By Proposition~\ref{prop:orthogonality-isbell-profiles},
  \[
    {}^\perp(\Omega_f^\perp)
    ={}^\perp\Omega_{M^*f}
    =\Omega_{M_*M^*f}.
  \]
  Thus \(\Omega_f\) is a left type if and only if
  \(\Omega_f=\Omega_{M_*M^*f}\), and this is equivalent to
  \(f=M_*M^*f\).  The right-hand
  statement is parallel.
\end{proof}

\begin{theorem}[Nucleus--types correspondence]\label{thm:nucleus-types}
  The assignments
  \[
    (f,g)\longmapsto(\Omega_f,\Omega_g),
    \qquad
    (A,B)\longmapsto(\lambda_A,\rho_B)
  \]
  are inverse bijections between the objects of \(\Nuc(M)\) and the paired
  types \(\mathrm{PType}(M)\).  Explicitly,
  \[
    (f,g)\in\Nuc(M)
    \quad\Longleftrightarrow\quad
    \Omega_f={}^\perp\Omega_g
    \text{ and }
    \Omega_g=\Omega_f^\perp.
  \]
\end{theorem}

\begin{proof}
  Suppose first that \((f,g)\in\Nuc(M)\).  Then \(g=M^*f\) and \(f=M_*g\).
  Proposition~\ref{prop:orthogonality-isbell-profiles} gives
  \[
    \Omega_f^\perp=\Omega_{M^*f}=\Omega_g,
    \qquad
    {}^\perp\Omega_g=\Omega_{M_*g}=\Omega_f.
  \]
  Hence \((\Omega_f,\Omega_g)\) is a paired type.

  Conversely, let \((A,B)\in\mathrm{PType}(M)\).  Thus
  \(B=A^\perp\) and \(A={}^\perp B\).  By
  Proposition~\ref{prop:orthogonality-isbell-profiles},
  \[
    B=A^\perp=\Omega_{M^*\lambda_A},
    \qquad
    A={}^\perp B=\Omega_{M_*\rho_B}.
  \]
  Applying profiles and using Corollary~\ref{cor:types-profile-complete}, gives
  \[
    \rho_B=M^*\lambda_A,
    \qquad
    \lambda_A=M_*\rho_B.
  \]
  Therefore \((\lambda_A,\rho_B)\in\Nuc(M)\). Finally, \(\varphi_{\Omega_f}=f\) for every coordinate \(f\), and paired
  types are recovered from their profiles by
  Corollary~\ref{cor:types-profile-complete}.
  Thus the two displayed assignments are inverse to one another.
\end{proof}

This completes the binary comparison.  Sections~\ref{sec:execution-types}
through~\ref{sec:categorical-structure} keep the same realisability datum but
ask a different question: how the execution product interacts with these
closed objects.

\section{Extending the execution to types}\label{sec:execution-types}

By assumption, the weighted execution product of Section~\ref{sec:types} is associative on
\(\cat C\times\R\).  We now ask how much of this multiplicative structure
descends to the binary left and right types generated by weighted
orthogonality.  The answer is limited: taking a raw product and
then closing it produces left and right types, but the resulting closed product
need not be associative.

Write
\[
  (a,\alpha)(b,\beta)=(ab,\alpha+\beta-M(a,b))
\]
for the weighted execution product.  If \(X,Y\subseteq\cat C\times\R\), define
their raw product by
\[
  XY:=\{xy\mid x\in X,\ y\in Y\}.
\]
Since weighted execution is associative, raw products of subsets are
associative:
\[
  (XY)Z=X(YZ).
\]

Recall the notation
\[
  \operatorname{cl}_L(X):={}^\perp(X^\perp),
  \qquad
  \operatorname{cl}_R(X):=({}^\perp X)^\perp
\]
for the left and right closures of Section~\ref{sec:types}.  Thus \(X\) is a
left type precisely when \(\operatorname{cl}_L(X)=X\), and a right type
precisely when \(\operatorname{cl}_R(X)=X\).

\begin{definition}\label{def:closed-products}
  For \(X,Y\subseteq\cat C\times\R\), define the \emph{left closed product}
  and \emph{right closed product} by
  \[
    X\odot_LY:=\operatorname{cl}_L(XY),
    \qquad
    X\odot_RY:=\operatorname{cl}_R(XY).
  \]
\end{definition}

Thus \(X\odot_LY\) is always a left type and \(X\odot_RY\) is always a right
type, by the orthogonal-complement calculus of Section~\ref{sec:types}.  First multiply, then apply the
appropriate binary closure; the notation keeps track of which closure is used.

\paragraph{Profile form of closed products.}
We shall also use the profile coordinates of
Section~\ref{sec:correspondence} for the closed products.  For any weighted
subset \(X\subseteq\cat C\times\R\), Proposition~\ref{prop:orthogonality-isbell-profiles}
gives, after applying complements twice,
\[
  \operatorname{cl}_L(X)=\Omega_{M_*M^*\varphi_X},
  \qquad
  \operatorname{cl}_R(X)=\Omega_{M^*M_*\varphi_X}.
\]

Suppose \(A\) and \(B\) are left types.  By
Corollary~\ref{cor:types-profile-complete},
\[
  A=\Omega_{\lambda_A},
  \qquad
  B=\Omega_{\lambda_B}.
\]
Let
\[
  \pi_{A,B}:=\varphi_{AB}
\]
be the profile of their raw product.  Unwinding the weighted execution product,
an element of \(AB\) lying over \(c\in\cat C\) has the form
\[
  (a,\alpha)(b,\beta)
  =
  (c,\alpha+\beta-M(a,b)),
  \qquad ab=c,
\]
with \(\alpha\le\lambda_A(a)\) and \(\beta\le\lambda_B(b)\).  Hence
\[
  \pi_{A,B}(c)
  =
  \sup_{ab=c}
  \bigl(\lambda_A(a)+\lambda_B(b)-M(a,b)\bigr).
\]
The supremum is taken in \(\Rbar\), with \(\sup\varnothing=-\infty\).  In a
finite unital monoid this supremum is a maximum, since every \(c\) has at least
the decompositions \(ec=c=ce\).

Therefore
\[
  A\odot_LB
  =
  \Omega_{M_*M^*\pi_{A,B}},
  \qquad
  \lambda_{A\odot_LB}=M_*M^*\pi_{A,B}.
\]
For compactness, we sometimes write the same operation directly at the level of
left profiles: for left profiles \(\ell,\ell'\),
\[
  \pi_{\ell,\ell'}(c)
  :=
  \sup_{ab=c}\bigl(\ell(a)+\ell'(b)-M(a,b)\bigr),
  \qquad
  \ell\odot_L\ell':=M_*M^*\pi_{\ell,\ell'}.
\]

Similarly, for right profiles \(\rho,\rho'\), we write
\[
  \rho\odot_R\rho':=M^*M_*\pi_{\rho,\rho'},
\]
where \(\pi_{\rho,\rho'}\) is defined by the same raw-product formula.

\subsection{The one-sided associativity obstruction}

The trefoil identity gives the basic rule for moving a weighted product across
orthogonality.

\begin{lemma}\label{lem.experp}
  For \(x=(a,\alpha)\), \(y=(b,\beta)\), and \(z=(c,\gamma)\) in
  \(\cat C\times\R\),
  \[
    xy\perp z
    \quad\Longleftrightarrow\quad
    x\perp yz .
  \]
\end{lemma}

\begin{proof}
  The condition \(xy\perp z\) is
  \[
    \alpha+\beta-M(a,b)+\gamma\le M(ab,c),
  \]
  or equivalently
  \[
    \alpha+\beta+\gamma\le M(a,b)+M(ab,c).
  \]
  By the trefoil identity \eqref{trefoil}, the right-hand side is
  \(M(b,c)+M(a,bc)\).  This is exactly the condition
  \[
    \alpha+\beta+\gamma-M(b,c)\le M(a,bc),
  \]
  which says \(x\perp yz\).
\end{proof}

The next identities isolate the asymmetry that will obstruct associativity.

\begin{proposition}\label{prop:left-closed-product-basic}
  For all \(X,Y\subseteq\cat C\times\R\),
  \[
    \operatorname{cl}_L(X)\odot_LY=X\odot_LY
  \]
  and
  \[
    X\odot_LY\subseteq X\odot_L\operatorname{cl}_L(Y).
  \]
\end{proposition}

\begin{proof}
  We prove the first identity by comparing right complements.  For
  \(z\in\cat C\times\R\),
  \begin{align*}
    z\in(\operatorname{cl}_L(X)Y)^\perp
     & \Longleftrightarrow
    ay\perp z
    \text{ for all }a\in\operatorname{cl}_L(X),\ y\in Y \\
     & \Longleftrightarrow
    a\perp yz
    \text{ for all }a\in\operatorname{cl}_L(X),\ y\in Y \\
     & \Longleftrightarrow
    yz\in(\operatorname{cl}_L(X))^\perp
    \text{ for all }y\in Y.
  \end{align*}
  By triple-complement cancellation,
  \((\operatorname{cl}_L(X))^\perp=X^\perp\).  Therefore the last condition is
  equivalent to
  \[
    x\perp yz
    \text{ for all }x\in X,\ y\in Y,
  \]
  which, by Lemma~\ref{lem.experp}, is equivalent to
  \(xy\perp z\) for all \(x\in X\) and \(y\in Y\).  Hence
  \[
    (\operatorname{cl}_L(X)Y)^\perp=(XY)^\perp.
  \]
  Applying the left complement gives
  \[
    \operatorname{cl}_L(\operatorname{cl}_L(X)Y)
    =\operatorname{cl}_L(XY),
  \]
  which is the first identity. For the second assertion, \(Y\subseteq\operatorname{cl}_L(Y)\), hence
  \(XY\subseteq X\operatorname{cl}_L(Y)\).  Taking right complements reverses
  the inclusion, and taking left complements reverses it again:
  \[
    \operatorname{cl}_L(XY)
    \subseteq
    \operatorname{cl}_L\bigl(X\operatorname{cl}_L(Y)\bigr).
  \]
  This is the desired containment.
\end{proof}

\begin{lemma}\label{lem.associntermed}
  For all \(X,Y,Z\subseteq\cat C\times\R\),
  \[
    \bigl((X\odot_LY)Z\bigr)^\perp=(X(YZ))^\perp.
  \]
\end{lemma}

\begin{proof}
  Let \(t\in\cat C\times\R\).  Then \(t\in((X\odot_LY)Z)^\perp\) if and only
  if \(dz\perp t\) for every \(d\in X\odot_LY\) and every \(z\in Z\).  By
  Lemma~\ref{lem.experp}, this is equivalent to requiring
  \(d\perp zt\) for every such \(d\) and \(z\).  Since
  \[
    (X\odot_LY)^\perp
    =(\operatorname{cl}_L(XY))^\perp
    =(XY)^\perp
  \]
  by triple-complement cancellation, the preceding condition is equivalent to
  \(zt\in(XY)^\perp\) for every \(z\in Z\).  Unwinding this, we get
  \(xy\perp zt\) for all \(x\in X\), \(y\in Y\), and \(z\in Z\).  Applying
  Lemma~\ref{lem.experp} once more, this is equivalent to
  \(x(yz)\perp t\) for all \(x,y,z\), namely \(t\in(X(YZ))^\perp\).
\end{proof}

\begin{proposition}\label{prop:left-associativity-obstruction}
  For all \(X,Y,Z\subseteq\cat C\times\R\),
  \[
    (X\odot_LY)\odot_LZ=X\odot_L(YZ)
  \]
  and therefore
  \[
    (X\odot_LY)\odot_LZ\subseteq X\odot_L(Y\odot_LZ).
  \]
  In particular, for a fixed triple \(X,Y,Z\), associativity of the left closed
  product is equivalent to the equality
  \[
    X\odot_L(YZ)=X\odot_L(Y\odot_LZ).
  \]
\end{proposition}

\begin{proof}
  By Lemma~\ref{lem.associntermed},
  \[
    (X\odot_LY)\odot_LZ
    =
    {}^\perp\bigl(((X\odot_LY)Z)^\perp\bigr)
    =
    {}^\perp\bigl((X(YZ))^\perp\bigr)
    =
    X\odot_L(YZ).
  \]
  Since \(YZ\subseteq Y\odot_LZ\), we have
  \(X(YZ)\subseteq X(Y\odot_LZ)\).  Applying right and then left complements
  gives
  \[
    X\odot_L(YZ)\subseteq X\odot_L(Y\odot_LZ).
  \]
  The final assertion follows from the first displayed equality.
\end{proof}
The right closed product has a parallel obstruction with
the variance reversed.

\subsection{A finite non-associativity example}

\begin{example}\label{ex:left-product-nonassociativity}
  We now give a finite example in which the containment of
  Proposition~\ref{prop:left-associativity-obstruction} is strict. We'll return to this example for the later unit-profile and two-out-of-three
  sharpness calculations.  Let
  \[
    \cat C=\{e,a,c,d\}
  \]
  be the monoid with identity \(e\) and multiplication table
  \[
    \begin{array}{c|cccc}
      \cdot & e & a & c & d \\
      \hline
      e     & e & a & c & d \\
      a     & a & e & c & d \\
      c     & c & d & c & d \\
      d     & d & c & c & d
    \end{array}
  \]
  and let
  \[
    p(e)=0,\qquad p(a)=-3,\qquad p(c)=1,\qquad p(d)=4.
  \]
  The associated measurement \(M(u,v)=p(uv)-p(u)-p(v)\), with rows and columns
  ordered as \((e,a,c,d)\), is
  \[
    \begin{array}{c|rrrr}
      M & e & a & c  & d    \\
      \hline
      e & 0 & 0 & 0  & 0    \\
      a & 0 & 6 & 3  & 3    \\
      c & 0 & 6 & -1 & -1   \\
      d & 0 & 0 & -4 & -4 .
    \end{array}
  \]
  We work in the profile coordinates of Section~\ref{sec:correspondence}.  Thus
  a vector records a function \(\cat C\to\Rbar\) in the order
  \((e,a,c,d)\), and the vector \(\ell\) represents the weighted subset
  \(\Omega_\ell\).  Let
  \[
    S_c=(-\infty,-\infty,0,-\infty),
    \qquad
    S_a=(-\infty,0,-\infty,-\infty).
  \]
  The corresponding subsets are generated, fibrewise downward in the weight
  coordinate, by \((c,0)\) and \((a,0)\).  Their left closures have profiles
  \[
    \ell_c:=M_*M^*S_c=(-6,0,0,-6),
    \qquad
    \ell_a:=M_*M^*S_a=(-6,0,-4,-7).
  \]
  Equivalently, these are the left types
  \[
    A_c:=\Omega_{\ell_c},
    \qquad
    A_a:=\Omega_{\ell_a}.
  \]

  In the present finite unital monoid, the raw product profile used above is
  \[
    \pi_{\ell,\ell'}(c)
    =
    \max_{ab=c}\bigl(\ell(a)+\ell'(b)-M(a,b)\bigr),
  \]
  and the left closed product of profiles is
  \[
    \ell\odot_L\ell':=M_*M^*\pi_{\ell,\ell'}.
  \]
  Using the displayed matrix for \(M\), a direct calculation gives
  \[
    \ell_c\odot_L\ell_a=(-6,-3,-3,-6),
    \qquad
    \ell_a\odot_L\ell_c=(-6,-3,-3,-6).
  \]
  Continuing the same calculation,
  \[
    (\ell_c\odot_L\ell_a)\odot_L\ell_c=(-8,-2,-2,-8),
  \]
  whereas
  \[
    \ell_c\odot_L(\ell_a\odot_L\ell_c)=(-5,-2,-2,-5).
  \]
  These are the profiles of the corresponding left closed products of
  \(A_c\) and \(A_a\).  Since inclusion between subsets of the form
  \(\Omega_\ell\) is equivalent to pointwise comparison of profiles, the first
  profile is strictly smaller than the second.  Hence
  \[
    (A_c\odot_LA_a)\odot_LA_c
    \subsetneq
    A_c\odot_L(A_a\odot_LA_c),
  \]
  so the left closed product is not associative, even for left types obtained by
  closing principal weighted elements.

\end{example}

\section{Middle types}\label{sec:middle}

Section~\ref{sec:execution-types} shows that the one-sided closed products of
binary left and right types need not be associative, even though the execution on $\cat{C} \times \R$ they extend is associative. What is missing from a binary left or right type is information
about an element placed between a left context and a right context.  This
section introduces the corresponding ternary orthogonality relation and its
complements.

\subsection{The ternary measurement}

For \(x,b,z\in\cat C\), define
\[
  M_3(x,b,z):=p(xbz)-p(x)-p(b)-p(z).
\]
This is the measurement of a middle element \(b\) placed between a left context
\(x\) and a right context \(z\).  It is related to the binary measurement by
the two splittings
\begin{equation}\label{eq:M3-splittings}
  M_3(x,b,z)
  =
  M(x,bz)+M(b,z)
  =
  M(xb,z)+M(x,b).
\end{equation}
Both identities follow immediately by expanding the definitions and using
associativity of execution. For weighted elements
\[
  x=(x_0,\xi),\qquad b=(b_0,\beta),\qquad z=(z_0,\zeta)
  \quad\text{in }\cat C\times\R,
\]
define \emph{middle/peripheral orthogonality} by
\begin{equation}\label{eq:middle-peripheral-orthogonality}
  b\Perp(x,z)
  \quad\Longleftrightarrow\quad
  \xi+\beta+\zeta\le M_3(x_0,b_0,z_0).
\end{equation}
Thus a single weighted element is tested in the middle position against an
ordered pair of peripheral weighted elements.

\begin{lemma}\label{lem:middle-orthogonality-binary-forms}
  For \(x,b,z\in\cat C\times\R\),
  \[
    b\Perp(x,z)
    \quad\Longleftrightarrow\quad
    xb\perp z
    \quad\Longleftrightarrow\quad
    x\perp bz .
  \]
\end{lemma}

\begin{proof}
  Write \(x=(x_0,\xi)\), \(b=(b_0,\beta)\), and \(z=(z_0,\zeta)\).  The
  condition \(xb\perp z\) is
  \[
    \xi+\beta-M(x_0,b_0)+\zeta\le M(x_0b_0,z_0),
  \]
  or equivalently
  \[
    \xi+\beta+\zeta\le M(x_0,b_0)+M(x_0b_0,z_0).
  \]
  By \eqref{eq:M3-splittings}, the right-hand side is
  \(M_3(x_0,b_0,z_0)\).  This is exactly
  \(b\Perp(x,z)\).  The equivalence with \(x\perp bz\) is the same calculation
  using the other splitting in \eqref{eq:M3-splittings}.
\end{proof}

\begin{lemma}\label{lem.Perpmove}
  For all \(a,a',x,z\in\cat C\times\R\),
  \[
    aa'\Perp(x,z)
    \quad\Longleftrightarrow\quad
    a\Perp(x,a'z)
    \quad\Longleftrightarrow\quad
    a'\Perp(xa,z).
  \]
\end{lemma}

\begin{proof}
  Using Lemma~\ref{lem:middle-orthogonality-binary-forms} and the execution
  shift Lemma~\ref{lem.experp},
  \[
    aa'\Perp(x,z)
    \Longleftrightarrow
    x(aa')\perp z
    \Longleftrightarrow
    xa\perp a'z
    \Longleftrightarrow
    a\Perp(x,a'z).
  \]
  The same starting condition is also equivalent to
  \[
    x(aa')\perp z
    \Longleftrightarrow
    (xa)a'\perp z
    \Longleftrightarrow
    a'\Perp(xa,z),
  \]
  where the middle equivalence uses associativity of weighted execution.
\end{proof}

\subsection{Middle and peripheral types}

We can associate two complement operations to the ternary measurement. For \(A\subseteq\cat C\times\R\), define its peripheral complement by
\[
  A^\Perp
  :=
  \{(x,z)\in(\cat C\times\R)^2
  \mid a\Perp(x,z)\text{ for all }a\in A\}.
\]
For \(P\subseteq(\cat C\times\R)^2\), define its middle complement by
\[
  {}^\Perp P
  :=
  \{a\in\cat C\times\R
  \mid a\Perp(x,z)\text{ for all }(x,z)\in P\}.
\]

\begin{definition}\label{def:middle-peripheral-closures}
  For \(A\subseteq\cat C\times\R\) and
  \(P\subseteq(\cat C\times\R)^2\), set
  \[
    \operatorname{cl}_{\mathrm{mid}}(A):={}^\Perp(A^\Perp),
    \qquad
    \operatorname{cl}_{\mathrm{per}}(P):=({}^\Perp P)^\Perp.
  \]
\end{definition}

\begin{definition}\label{def:middle-peripheral-types}
  A \emph{middle type} is a subset \(A\subseteq\cat C\times\R\) such that
  \[
    \operatorname{cl}_{\mathrm{mid}}(A)=A.
  \]
  A \emph{peripheral type} is a subset
  \(P\subseteq(\cat C\times\R)^2\) such that
  \[
    \operatorname{cl}_{\mathrm{per}}(P)=P.
  \]
\end{definition}

Equivalently, middle types are the subsets of the form \({}^\Perp P\), and
peripheral types are the subsets of the form \(A^\Perp\).  The formal complement
calculus is the same as in Section~\ref{sec:types}; we describe it here in the middle-peripheral situation for ease of reference.

\begin{proposition}[Middle complement calculus]\label{prop:middle-polarity-calculus}
  Let \(A,A'\subseteq\cat C\times\R\) and
  \(P,P'\subseteq(\cat C\times\R)^2\).
  \begin{enumerate}[label=\textup{(\roman*)}]
    \item If \(A\subseteq A'\), then \((A')^\Perp\subseteq A^\Perp\).  If
          \(P\subseteq P'\), then \({}^\Perp P'\subseteq{}^\Perp P\).

    \item One has
          \[
            A\subseteq{}^\Perp(A^\Perp),
            \qquad
            P\subseteq({}^\Perp P)^\Perp.
          \]

    \item Alternating three complements cancels to one complement:
          \[
            {}^\Perp\bigl(({}^\Perp P)^\Perp\bigr)={}^\Perp P,
            \qquad
            \bigl({}^\Perp(A^\Perp)\bigr)^\Perp=A^\Perp.
          \]

    \item Closure does not change the opposite complement:
          \[
            (\operatorname{cl}_{\mathrm{mid}}A)^\Perp=A^\Perp,
            \qquad
            {}^\Perp(\operatorname{cl}_{\mathrm{per}}P)={}^\Perp P.
          \]
  \end{enumerate}
\end{proposition}

\begin{proof}
  These are the standard identities for antitone orthogonal-complement
  operations.  Antitonicity follows directly from the definitions: enlarging the
  set being tested against imposes more inequalities.  The double-complement
  containments follow because every element is orthogonal to all elements in its
  own complement.  Applying antitonicity to these containments gives the reverse
  inclusions needed for triple-complement cancellation.  The final identities are
  exactly those
  triple-complement identities applied to the closures in
  Definition~\ref{def:middle-peripheral-closures}.
\end{proof}

\subsection{Middle profiles}

The ternary measurement may also be regarded as an \(\Rbar\)-profunctor
\[
  \Mtc\colon \cat C\nrightarrow\cat C\times\cat C,
  \qquad
  \Mtc\bigl(b,(x,z)\bigr)=M_3(x,b,z).
\]
We now record the profile form of the middle--peripheral types.

For a middle subset \(A\subseteq\cat C\times\R\), write
\[
  \mu_A:=\varphi_A\colon\cat C\to\Rbar
\]
for its ordinary fibrewise profile.  For a peripheral subset
\(P\subseteq(\cat C\times\R)^2\), define its boundary profile by
\[
  \kappa_P(x,z)
  :=
  \sup\{\xi+\zeta\in\R
  \mid ((x,\xi),(z,\zeta))\in P\}.
\]
Conversely, for a function \(h\colon\cat C\times\cat C\to\Rbar\), set
\[
  \Omega^\partial_h
  :=
  \{((x,\xi),(z,\zeta))\in(\cat C\times\R)^2
  \mid \xi+\zeta\le h(x,z)\}.
\]
Thus \(\Omega^\partial_h\) is the peripheral subset recovered from the
two-variable boundary profile \(h\).

The Isbell conjugates for \(\Mtc\) are
\[
  (\Mtc^*\mu)(x,z)
  =
  \inf_{b\in\cat C}\bigl(M_3(x,b,z)-\mu(b)\bigr),
\]
and
\[
  (\Mtc_*\kappa)(b)
  =
  \inf_{x,z\in\cat C}\bigl(M_3(x,b,z)-\kappa(x,z)\bigr).
\]

\begin{proposition}[Middle orthogonality in profile coordinates]
  \label{prop:middle-orthogonality-profile}
  For \(A\subseteq\cat C\times\R\) and
  \(P\subseteq(\cat C\times\R)^2\),
  \[
    A^\Perp=\Omega^\partial_{\Mtc^*\mu_A},
    \qquad
    {}^\Perp P=\Omega_{\Mtc_*\kappa_P}.
  \]
\end{proposition}

\begin{proof}
  Let \(((x,\xi),(z,\zeta))\in(\cat C\times\R)^2\).  Then
  \(((x,\xi),(z,\zeta))\in A^\Perp\) if and only if
  \[
    \xi+\beta+\zeta\le M_3(x,b,z)
  \]
  for every \((b,\beta)\in A\).  For each fixed \(b\), this is equivalent to
  \[
    \xi+\zeta\le M_3(x,b,z)-\mu_A(b).
  \]
  Requiring this for every \(b\) gives
  \[
    \xi+\zeta
    \le
    \inf_{b\in\cat C}\bigl(M_3(x,b,z)-\mu_A(b)\bigr)
    =
    (\Mtc^*\mu_A)(x,z),
  \]
  which is precisely membership in \(\Omega^\partial_{\Mtc^*\mu_A}\).

  The second identity is the same calculation in the other direction.  A
  weighted element \((b,\beta)\) lies in \({}^\Perp P\) if and only if
  \[
    \beta+\xi+\zeta\le M_3(x,b,z)
  \]
  for every \(((x,\xi),(z,\zeta))\in P\).  For each pair \((x,z)\), this is
  equivalent to
  \[
    \beta\le M_3(x,b,z)-\kappa_P(x,z).
  \]
  Requiring this for all \(x,z\) gives
  \[
    \beta
    \le
    \inf_{x,z\in\cat C}
    \bigl(M_3(x,b,z)-\kappa_P(x,z)\bigr)
    =
    (\Mtc_*\kappa_P)(b),
  \]
  as required.
\end{proof}

\begin{corollary}
  \label{cor:middle-types-profile}
  Every middle type and every peripheral type is recovered from its profile.
  More precisely, if \(A\) is a middle type and \(P\) is a peripheral type, then
  \[
    A=\Omega_{\mu_A},
    \qquad
    P=\Omega^\partial_{\kappa_P}.
  \]
  Moreover, for a coordinate \(\mu\colon\cat C\to\Rbar\),
  \[
    \Omega_\mu\text{ is a middle type}
    \quad\Longleftrightarrow\quad
    \Mtc_*\Mtc^*\mu=\mu,
  \]
  and for a boundary coordinate
  \(\kappa\colon\cat C\times\cat C\to\Rbar\),
  \[
    \Omega^\partial_\kappa\text{ is a peripheral type}
    \quad\Longleftrightarrow\quad
    \Mtc^*\Mtc_*\kappa=\kappa.
  \]
  Thus a middle type \(A\), together with its peripheral complement, is
  represented by the nuclear point
  \[
    (\mu_A,\Mtc^*\mu_A)\in\Nuc(\Mtc).
  \]
\end{corollary}

\subsection{The middle product}

If \(A,B\subseteq\cat C\times\R\), let \(AB\) denote their raw product under
weighted execution:
\[
  AB:=\{ab\mid a\in A,\ b\in B\}.
\]
The \emph{middle closed product} is
\begin{equation}\label{eq:middle-product}
  A\odot B:=\operatorname{cl}_{\mathrm{mid}}(AB)
  ={}^\Perp((AB)^\Perp).
\end{equation}
Thus \(A\odot B\) is a middle type for arbitrary subsets \(A\) and \(B\), and
in particular for middle types.

\begin{lemma}\label{lem:middle-product-complement}
  For all \(A,B,C\subseteq\cat C\times\R\),
  \[
    ((A\odot B)C)^\Perp=(A(BC))^\Perp
  \]
  and
  \[
    (A(B\odot C))^\Perp=(A(BC))^\Perp.
  \]
\end{lemma}

\begin{proof}
  Let \((x,z)\in(\cat C\times\R)^2\).  Then
  \((x,z)\in((A\odot B)C)^\Perp\) if and only if
  \(dc\Perp(x,z)\) for every \(d\in A\odot B\) and \(c\in C\).  By
  Lemma~\ref{lem.Perpmove}, this is equivalent to
  \(d\Perp(x,cz)\) for every such \(d\) and \(c\).  Since
  \[
    (A\odot B)^\Perp
    =
    (\operatorname{cl}_{\mathrm{mid}}(AB))^\Perp
    =
    (AB)^\Perp,
  \]
  the preceding condition is equivalent to requiring
  \((x,cz)\in(AB)^\Perp\) for every \(c\in C\).  Unwinding again, this says
  \(ab\Perp(x,cz)\) for all \(a\in A\), \(b\in B\), and \(c\in C\).  By
  Lemma~\ref{lem.Perpmove}, this is equivalent to
  \(a(bc)\Perp(x,z)\) for all \(a,b,c\), which is precisely
  \((x,z)\in(A(BC))^\Perp\).

  The second identity is parallel.  The condition
  \((x,z)\in(A(B\odot C))^\Perp\) is equivalent to
  \(ae\Perp(x,z)\) for every \(a\in A\) and \(e\in B\odot C\), hence to
  \(e\Perp(xa,z)\) for every such \(a,e\).  Since
  \((B\odot C)^\Perp=(BC)^\Perp\), this is equivalent to
  \(bc\Perp(xa,z)\) for all \(b\in B,c\in C\).  Applying
  Lemma~\ref{lem.Perpmove} once more gives \(a(bc)\Perp(x,z)\), as required.
\end{proof}

\begin{proposition}[Associativity of the middle product]\label{thm.midassoc}
  For all \(A,B,C\subseteq\cat C\times\R\),
  \[
    (A\odot B)\odot C=A\odot(B\odot C).
  \]
  In particular, the middle closed product restricts to an associative product
  on middle types.
\end{proposition}

\begin{proof}
  By Lemma~\ref{lem:middle-product-complement},
  \[
    ((A\odot B)C)^\Perp=(A(BC))^\Perp=(A(B\odot C))^\Perp.
  \]
  Applying the middle complement \({}^\Perp(-)\) to the two outer terms gives
  \[
    (A\odot B)\odot C
    =
    {}^\Perp(((A\odot B)C)^\Perp)
    =
    {}^\Perp((A(B\odot C))^\Perp)
    =
    A\odot(B\odot C).
  \]
\end{proof}

\begin{remark}
  The same raw product also records edge actions with the binary types of
  Section~\ref{sec:types}.  If \(L\) is a left type and \(A\) is a middle type,
  the left edge action is the Section~\ref{sec:execution-types} closed product
  \[
    L\odot_L A:=\operatorname{cl}_L(LA).
  \]
  If \(A\) is a middle type and \(R\) is a right type, the right edge action is
  \[
    A\odot_R R:=\operatorname{cl}_R(AR).
  \]
  These actions record how a middle type meets one binary boundary type.  Unit or
  residual laws for them require the additional closure and variance hypotheses
  stated in the next section.
\end{remark}

Section~\ref{sec:categorical-structure} adds the extra structure needed to
discuss units and residual operations for the middle product.

\section{Units and residual operations}\label{sec:categorical-structure}

Having constructed the associative middle product from the ternary
orthogonality relation, we now record the additional algebraic structure that
can be used later in the coordinate calculations: a unit when the execution
product has one, and the two oriented residual operations for the middle
product.

\subsection{Units}

The original realisability datum did not include a unit.  For the unit
statements in this section, assume that the execution product on \(\cat C\) has
a two-sided unit \(\epsilon\).  Set
\[
  \mathbf e:=(\epsilon,-p(\epsilon))\in\cat C\times\R
\]
and define the middle type generated by this weighted element by
\[
  \mathbf 1
  :=
  \operatorname{cl}_{\mathrm{mid}}(\{\mathbf e\})
  =
  {}^\Perp(\{\mathbf e\}^\Perp).
\]

\begin{lemma}\label{lem:weighted-unit}
  The weighted element \(\mathbf e\) is a two-sided unit for weighted
  execution on \(\cat C\times\R\).
\end{lemma}

\begin{proof}
  Let \(x=(a,\alpha)\).  Since \(a\epsilon=\epsilon a=a\),
  \begin{align*}
    M(a,\epsilon) & =p(a)-p(a)-p(\epsilon)=-p(\epsilon),
    \\
    M(\epsilon,a) & =p(a)-p(\epsilon)-p(a)=-p(\epsilon).
  \end{align*}
  Hence
  \[
    x\mathbf e
    =
    (a,\alpha)(\epsilon,-p(\epsilon))
    =
    (a,\alpha),
  \]
  and similarly \(\mathbf e x=x\).
\end{proof}

\begin{proposition}\label{prop:middle-unit}
  Under the unit hypothesis above, \(\mathbf 1\) is a two-sided unit for the
  middle product.  For every middle type \(A\),
  \[
    \mathbf 1\odot A=A=A\odot\mathbf 1.
  \]
\end{proposition}

\begin{proof}
  Since \(\{\mathbf e\}A=A=A\{\mathbf e\}\) by
  Lemma~\ref{lem:weighted-unit}, we have
  \[
    \{\mathbf e\}\odot A
    =
    \operatorname{cl}_{\mathrm{mid}}(A)
    =
    A
  \]
  and similarly \(A\odot\{\mathbf e\}=A\).  Since
  \(\{\mathbf e\}\{\mathbf e\}=\{\mathbf e\}\), we have
  \(\mathbf 1=\{\mathbf e\}\odot\{\mathbf e\}\).  Proposition~\ref{thm.midassoc}
  then gives
  \[
    \mathbf 1\odot A
    =
    (\{\mathbf e\}\odot\{\mathbf e\})\odot A
    =
    \{\mathbf e\}\odot(\{\mathbf e\}\odot A)
    =
    A,
  \]
  and the right unit identity is analogous.
\end{proof}

\subsection{Left and right units}

The same weighted unit also generates one-sided units for the binary closed
products of Section~\ref{sec:execution-types}.  Define
\[
  \mathbf 1_L:=\operatorname{cl}_L(\{\mathbf e\}),
  \qquad
  \mathbf 1_R:=\operatorname{cl}_R(\{\mathbf e\}).
\]
Thus \(\mathbf 1_L\) is a left type and \(\mathbf 1_R\) is a right type.
Their profiles make visible the difference between testing the weighted unit
against one boundary and testing it against both boundaries.

\begin{proposition}[Profiles of the unit types]
  \label{prop:left-right-unit-profiles}
  Let
  \[
    \lambda_{\mathbf 1_L},
    \qquad
    \rho_{\mathbf 1_R},
    \qquad
    \mu_{\mathbf 1}
  \]
  denote the left, right, and middle profiles of \(\mathbf 1_L\),
  \(\mathbf 1_R\), and \(\mathbf 1\), respectively.  As \(\Rbar\)-valued
  functions,
  \begin{align*}
    \lambda_{\mathbf 1_L}(b)
     & =
    \inf_{z\in\cat C}M(b,z),
    \\
    \rho_{\mathbf 1_R}(b)
     & =
    \inf_{x\in\cat C}M(x,b),
    \\
    \mu_{\mathbf 1}(b)
     & =
    \inf_{x,z\in\cat C}\bigl(M_3(x,b,z)-M(x,z)\bigr).
  \end{align*}
  Equivalently, the term in the third infimum is
  \[
    M_3(x,b,z)-M(x,z)=p(xbz)-p(xz)-p(b).
  \]
\end{proposition}

\begin{proof}
  Since \(M(\epsilon,z)=-p(\epsilon)\), a weighted element \((z,\zeta)\) lies
  in \(\{\mathbf e\}^{\perp}\) if and only if
  \[
    -p(\epsilon)+\zeta\le -p(\epsilon),
  \]
  equivalently \(\zeta\le 0\).  Hence \(\{\mathbf e\}^{\perp}\) is the right
  type with constant profile \(0\).  Therefore \((b,\beta)\in\mathbf 1_L\) if
  and only if
  \[
    \beta+\zeta\le M(b,z)
    \quad\text{for all }z\in\cat C\text{ and all }\zeta\le 0.
  \]
  Since \(\zeta=0\) is allowed, this is equivalent to
  \[
    \beta\le\inf_{z\in\cat C}M(b,z),
  \]
  which gives the formula for \(\lambda_{\mathbf 1_L}\).

  The calculation for \(\mathbf 1_R\) is the same with the two variables
  interchanged.  Namely, \({}^{\perp}\{\mathbf e\}\) has constant profile \(0\),
  and taking its right complement gives
  \[
    \rho_{\mathbf 1_R}(b)=\inf_{x\in\cat C}M(x,b).
  \]

  For the middle unit, first observe that \(\{\mathbf e\}^{\Perp}\) is exactly
  the binary orthogonality relation on the two peripheral variables.  Indeed,
  \(M_3(x,\epsilon,z)=M(x,z)-p(\epsilon)\), so
  \[
    (\epsilon,-p(\epsilon))\Perp((x,\xi),(z,\zeta))
    \quad\Longleftrightarrow\quad
    \xi+\zeta\le M(x,z).
  \]
  Hence \((b,\beta)\in\mathbf 1={}^\Perp(\{\mathbf e\}^{\Perp})\) if and only
  if
  \[
    \beta+\xi+\zeta\le M_3(x,b,z)
  \]
  for all \(x,z\in\cat C\) and all \(\xi,\zeta\in\R\) with
  \(\xi+\zeta\le M(x,z)\).  Since \(\xi+\zeta=M(x,z)\) is allowed, this is
  equivalent to
  \[
    \beta\le
    \inf_{x,z\in\cat C}\bigl(M_3(x,b,z)-M(x,z)\bigr),
  \]
  which gives the formula for \(\mu_{\mathbf 1}\).

  The final displayed identity is obtained by expanding \(M_3(x,b,z)\) and
  \(M(x,z)\).
\end{proof}

\begin{proposition}[Comparison of unit profiles]
  \label{prop:unit-profile-comparison}
  Under the unit hypothesis,
  \[
    \mu_{\mathbf 1}(b)\le \lambda_{\mathbf 1_L}(b),
    \qquad
    \mu_{\mathbf 1}(b)\le \rho_{\mathbf 1_R}(b)
  \]
  for every \(b\in\cat C\).  Equivalently,
  \[
    \mu_{\mathbf 1}\le
    \min(\lambda_{\mathbf 1_L},\rho_{\mathbf 1_R})
  \]
  pointwise.
\end{proposition}

\begin{proof}
  By Proposition~\ref{prop:left-right-unit-profiles},
  \[
    \mu_{\mathbf 1}(b)
    =
    \inf_{x,z\in\cat C}\bigl(M_3(x,b,z)-M(x,z)\bigr).
  \]
  Restricting the infimum to \(x=\epsilon\) gives
  \[
    \mu_{\mathbf 1}(b)
    \le
    \inf_{z\in\cat C}\bigl(M_3(\epsilon,b,z)-M(\epsilon,z)\bigr)
    =
    \inf_{z\in\cat C}M(b,z)
    =
    \lambda_{\mathbf 1_L}(b).
  \]
  Restricting the infimum to \(z=\epsilon\) gives
  \[
    \mu_{\mathbf 1}(b)
    \le
    \inf_{x\in\cat C}\bigl(M_3(x,b,\epsilon)-M(x,\epsilon)\bigr)
    =
    \inf_{x\in\cat C}M(x,b)
    =
    \rho_{\mathbf 1_R}(b).
  \]
\end{proof}

\begin{proposition}[One-sided unit laws]\label{prop:left-right-unit-laws}
  Under the unit hypothesis, if \(A\) is a left type, then
  \[
    \mathbf 1_L\odot_L A=A.
  \]
  If \(B\) is a right type, then
  \[
    B\odot_R\mathbf 1_R=B.
  \]
\end{proposition}

\begin{proof}
  The left identity follows from Proposition~\ref{prop:left-closed-product-basic}
  and Lemma~\ref{lem:weighted-unit}:
  \[
    \mathbf 1_L\odot_L A
    =
    \operatorname{cl}_L(\{\mathbf e\})\odot_L A
    =
    \{\mathbf e\}\odot_L A
    =
    \operatorname{cl}_L(A)
    =
    A.
  \]

  For the right identity we use the right-hand analogue
  \[
    X\odot_R\operatorname{cl}_R(Y)=X\odot_RY.
  \]
  Indeed, comparing left complements, \(q\in{}^\perp(X\operatorname{cl}_R(Y))\)
  if and only if \(q\perp xy\) for all \(x\in X\) and
  \(y\in\operatorname{cl}_R(Y)\).  By Lemma~\ref{lem.experp}, this is
  equivalent to \(qx\perp y\) for all such \(x\) and \(y\), or
  \(qx\in{}^\perp\operatorname{cl}_R(Y)={}^\perp Y\) for all \(x\in X\).  This
  is equivalent, again by Lemma~\ref{lem.experp}, to
  \(q\in{}^\perp(XY)\).  Taking right complements gives the displayed identity.
  Hence, for a right type \(B\),
  \[
    B\odot_R\mathbf 1_R
    =
    B\odot_R\operatorname{cl}_R(\{\mathbf e\})
    =
    B\odot_R\{\mathbf e\}
    =
    \operatorname{cl}_R(B)
    =
    B.
  \]
\end{proof}

\begin{remark}
  The preceding proposition is deliberately one-sided.  It gives a left unit for
  the left closed product and a right unit for the right closed product.  The
  identities \(A\odot_L\mathbf 1_L=A\) and \(\mathbf 1_R\odot_RB=B\) would
  require control of closure in the second variable of \(\odot_L\), respectively
  the first variable of \(\odot_R\).  Section~\ref{sec:execution-types} isolates
  this one-sided closure issue.
\end{remark}

\begin{example}[The unit profiles in the running example]
  Return to the finite monoid of
  Example~\ref{ex:left-product-nonassociativity}.  Its unit is \(e\), and since
  \(p(e)=0\), the weighted unit is \(\mathbf e=(e,0)\).  From the displayed
  matrix for \(M\), the row minima and column minima give the one-sided unit
  profiles, and the middle unit formula gives the middle profile.  In the order
  \((e,a,c,d)\),
  \[
    \lambda_{\mathbf 1_L}=(0,0,-1,-4),
    \qquad
    \rho_{\mathbf 1_R}=(0,0,-4,-4),
    \qquad
    \mu_{\mathbf 1}=(0,0,-4,-7).
  \]
  Thus the running example distinguishes the left, right, and middle unit
  profiles.  It also shows that the comparison in
  Proposition~\ref{prop:unit-profile-comparison} can be strict: at \(d\),
  \[
    \mu_{\mathbf 1}(d)=-7
    <
    -4
    =
    \min(\lambda_{\mathbf 1_L}(d),\rho_{\mathbf 1_R}(d)).
  \]
  The value \(-7\) is attained in the middle unit formula by the peripheral
  pair \((x,z)=(c,a)\): the multiplication changes from \(ca=d\) to
  \(cda=c\).
\end{example}

\subsection{Residuals}

The middle product is not commutative, so it has two oriented residuals.  We
first define them as subsets of \(\cat C\times\R\).

\begin{definition}\label{def.lrimplications}
  Let \(A\subseteq\cat C\times\R\) and let \(B\) be a middle type.  Define
  \[
    A\multimap_l B
    :=
    \{r\in\cat C\times\R\mid ar\in B\text{ for all }a\in A\},
  \]
  and
  \[
    A\multimap_r B
    :=
    \{r\in\cat C\times\R\mid ra\in B\text{ for all }a\in A\}.
  \]
\end{definition}

For residuals of middle types we also need two actions of a middle subset on a
peripheral subset.  If \(A\subseteq\cat C\times\R\) and
\(P\subseteq(\cat C\times\R)^2\), set
\[
  A\odot_l P
  :=
  \{(x a,z)\mid a\in A,\ (x,z)\in P\},
\]
and
\[
  A\odot_r P
  :=
  \{(x,a z)\mid a\in A,\ (x,z)\in P\}.
\]
These are raw peripheral subsets; no peripheral closure is included in the
notation.

\begin{proposition}[Residuals as middle types]\label{thm.dotmultimap}
  If \(A\subseteq\cat C\times\R\) and \(B\) is a middle type, then
  \[
    A\multimap_l B={}^\Perp(A\odot_l B^\Perp),
    \qquad
    A\multimap_r B={}^\Perp(A\odot_r B^\Perp).
  \]
  In particular, \(A\multimap_l B\) and \(A\multimap_r B\) are middle types.
\end{proposition}

\begin{proof}
  We prove the left residual formula.  For \(r\in\cat C\times\R\),
  \[
    r\in A\multimap_l B
  \]
  means that \(ar\in B\) for every \(a\in A\).  Since \(B\) is a middle type,
  this is equivalent to
  \[
    ar\Perp(x,z)
    \quad
    \text{for all }a\in A\text{ and }(x,z)\in B^\Perp.
  \]
  By Lemma~\ref{lem.Perpmove}, the displayed condition is equivalent to
  \[
    r\Perp(xa,z)
    \quad
    \text{for all }a\in A\text{ and }(x,z)\in B^\Perp,
  \]
  which is precisely \(r\in{}^\Perp(A\odot_l B^\Perp)\).

  The proof of the right residual formula is the same: \(ra\Perp(x,z)\) is
  equivalent to \(r\Perp(x,az)\) by Lemma~\ref{lem.Perpmove}.
\end{proof}

\begin{proposition}[Residuation]\label{prop.adj}
  Let \(A,B,C\) be middle types.  Then
  \[
    A\odot B\subseteq C
    \quad\Longleftrightarrow\quad
    B\subseteq A\multimap_l C
    \quad\Longleftrightarrow\quad
    A\subseteq B\multimap_r C.
  \]
\end{proposition}

\begin{proof}
  Since \(C\) is middle closed,
  \[
    A\odot B=\operatorname{cl}_{\mathrm{mid}}(AB)\subseteq C
    \quad\Longleftrightarrow\quad
    AB\subseteq C.
  \]
  The latter condition says that \(ab\in C\) for all \(a\in A\) and
  \(b\in B\).  Equivalently, every \(b\in B\) lies in \(A\multimap_l C\), and
  equivalently every \(a\in A\) lies in \(B\multimap_r C\).
\end{proof}

\begin{proposition}[Partial transitivity]\label{thm.partialtransitivity}
  Let \(A,B,C\) be middle types.  Then
  \[
    (A\multimap_l B)\odot(B\multimap_l C)\subseteq A\multimap_l C
  \]
  and
  \[
    (B\multimap_r C)\odot(A\multimap_r B)\subseteq A\multimap_r C.
  \]
\end{proposition}

\begin{proof}
  We prove the left-handed statement.  Since \(A\multimap_l C\) is a middle
  type by Proposition~\ref{thm.dotmultimap}, it is enough to show that the raw
  product
  \[
    (A\multimap_l B)(B\multimap_l C)
  \]
  is contained in \(A\multimap_l C\).  Let \(r\in A\multimap_l B\) and
  \(s\in B\multimap_l C\).  For each \(a\in A\), one has \(ar\in B\), and
  therefore \((ar)s\in C\).  Associativity of weighted execution gives
  \(a(rs)=(ar)s\), so \(rs\in A\multimap_l C\).

  The right-handed statement is parallel.  If \(r\in B\multimap_r C\) and
  \(s\in A\multimap_r B\), then for \(a\in A\) one has \(sa\in B\), and hence
  \(r(sa)\in C\).  Associativity gives \((rs)a=r(sa)\), so
  \(rs\in A\multimap_r C\).
\end{proof}

\begin{lemma}\label{lem:unit-detects-inclusion}
  Assume the execution product has a unit, and let \(\mathbf e\) be the weighted
  unit of Lemma~\ref{lem:weighted-unit}.  If \(A\) and \(B\) are middle types,
  then
  \[
    A\subseteq B
    \quad\Longleftrightarrow\quad
    \mathbf e\in A\multimap_l B
    \quad\Longleftrightarrow\quad
    \mathbf e\in A\multimap_r B.
  \]
\end{lemma}

\begin{proof}
  This is immediate from \(a\mathbf e=a=\mathbf e a\) for every
  \(a\in\cat C\times\R\).
\end{proof}

\subsection{The assembled type calculus}
\label{subsec:assembled-type-calculus}

The constructions of Sections~\ref{sec:execution-types}--\ref{sec:categorical-structure}
are organized by the closure applied after raw weighted execution.  Let
\[
  \mathsf L=\fix(\operatorname{cl}_L),
  \qquad
  \mathsf M=\fix(\operatorname{cl}_{\mathrm{mid}}),
  \qquad
  \mathsf R=\fix(\operatorname{cl}_R)
\]
denote the left, middle, and right types.  For arbitrary
\(X,Y\subseteq\cat C\times\R\), the same raw product \(XY\) has three closed
forms:
\[
  X\odot_LY=\operatorname{cl}_L(XY)\in\mathsf L,
  \qquad
  X\odot Y=\operatorname{cl}_{\mathrm{mid}}(XY)\in\mathsf M,
  \qquad
  X\odot_RY=\operatorname{cl}_R(XY)\in\mathsf R.
\]
The product symbol keeps track of which closure is used.  The middle types carry
the associative core of the calculus.  Ordered by inclusion, they carry the
following Lambek-style structure.

\begin{proposition}[The assembled middle calculus]
  \label{prop:assembled-middle-calculus}
  Let \(A,B,C\) be middle types.
  \begin{enumerate}[label=\textup{(\roman*)}]
    \item The product \(\odot\) is monotone in both variables and associative.  If
          the execution product has a unit, then
          \[
            \mathbf 1\odot A=A=A\odot\mathbf 1 .
          \]

    \item The residuals are middle types and are characterized by
          \[
            A\odot B\subseteq C
            \quad\Longleftrightarrow\quad
            B\subseteq A\multimap_l C
            \quad\Longleftrightarrow\quad
            A\subseteq B\multimap_r C .
          \]

    \item The residuals curry the product:
          \[
            (A\odot B)\multimap_l C
            =
            B\multimap_l(A\multimap_l C),
          \]
          and
          \[
            (A\odot B)\multimap_r C
            =
            A\multimap_r(B\multimap_r C).
          \]

    \item The two orientations satisfy the mixed identity
          \[
            A\multimap_l(B\multimap_r C)
            =
            B\multimap_r(A\multimap_l C).
          \]

    \item Residual arrows compose by the inclusions
          \[
            (A\multimap_l B)\odot(B\multimap_l C)
            \subseteq
            A\multimap_l C,
          \]
          and
          \[
            (B\multimap_r C)\odot(A\multimap_r B)
            \subseteq
            A\multimap_r C.
          \]

    \item Under the unit hypothesis, inclusion is detected by the middle unit:
          \[
            A\subseteq B
            \quad\Longleftrightarrow\quad
            \mathbf 1\subseteq A\multimap_l B
            \quad\Longleftrightarrow\quad
            \mathbf 1\subseteq A\multimap_r B .
          \]
  \end{enumerate}
\end{proposition}

\begin{proof}
  Monotonicity follows from monotonicity of raw product and of middle closure.
  Associativity is Proposition~\ref{thm.midassoc}, the unit law is
  Proposition~\ref{prop:middle-unit}, the fact that residuals are middle types is
  Proposition~\ref{thm.dotmultimap}, and the residual characterization is
  Proposition~\ref{prop.adj}.

  The currying identities and the mixed identity are formal consequences of
  associativity and residuation.  Let \(D\) be a middle type.  Then
  \[
    \begin{aligned}
      D\subseteq (A\odot B)\multimap_l C
       & \Longleftrightarrow
      (A\odot B)\odot D\subseteq C     \\
       & \Longleftrightarrow
      A\odot(B\odot D)\subseteq C      \\
       & \Longleftrightarrow
      B\odot D\subseteq A\multimap_l C \\
       & \Longleftrightarrow
      D\subseteq B\multimap_l(A\multimap_l C).
    \end{aligned}
  \]
  Similarly,
  \[
    \begin{aligned}
      D\subseteq (A\odot B)\multimap_r C
       & \Longleftrightarrow
      D\odot(A\odot B)\subseteq C      \\
       & \Longleftrightarrow
      (D\odot A)\odot B\subseteq C     \\
       & \Longleftrightarrow
      D\odot A\subseteq B\multimap_r C \\
       & \Longleftrightarrow
      D\subseteq A\multimap_r(B\multimap_r C),
    \end{aligned}
  \]
  and
  \[
    \begin{aligned}
      D\subseteq A\multimap_l(B\multimap_r C)
       & \Longleftrightarrow
      A\odot D\subseteq B\multimap_r C \\
       & \Longleftrightarrow
      (A\odot D)\odot B\subseteq C     \\
       & \Longleftrightarrow
      A\odot(D\odot B)\subseteq C      \\
       & \Longleftrightarrow
      D\odot B\subseteq A\multimap_l C \\
       & \Longleftrightarrow
      D\subseteq B\multimap_r(A\multimap_l C).
    \end{aligned}
  \]
  Since the two sides compared in each displayed identity are middle types,
  testing against all middle \(D\) gives equality.  The composition inclusions
  are Proposition~\ref{thm.partialtransitivity}.  Finally,
  Lemma~\ref{lem:unit-detects-inclusion} detects inclusion by membership of the
  weighted unit \(\mathbf e\) in either residual; because the residuals are
  middle types, this is equivalent to containing
  \(\operatorname{cl}_{\mathrm{mid}}(\{\mathbf e\})=\mathbf 1\).
\end{proof}

The boundary products sit alongside this middle calculus as one-sided boundary
forms of the same execution product.  If \(X_L\in\mathsf L\),
\(A\in\mathsf M\), and \(X_R\in\mathsf R\), then the raw products can be closed
in the middle sort,
\[
  X_L\odot A=\operatorname{cl}_{\mathrm{mid}}(X_LA),
  \qquad
  A\odot X_R=\operatorname{cl}_{\mathrm{mid}}(AX_R),
  \qquad
  X_L\odot X_R=\operatorname{cl}_{\mathrm{mid}}(X_LX_R),
\]
or at the boundary,
\[
  X_L\odot_L A=\operatorname{cl}_L(X_LA),
  \qquad
  A\odot_R X_R=\operatorname{cl}_R(AX_R).
\]
Under the unit hypothesis, the boundary unit laws proved above are
\[
  \mathbf 1_L\odot_L X_L=X_L,
  \qquad
  X_R\odot_R\mathbf 1_R=X_R.
\]
Further boundary identities are governed by the containment
\[
  X\odot_L(YZ)\subseteq X\odot_L(Y\odot_L Z)
\]
from Proposition~\ref{prop:left-associativity-obstruction}, together with the
right-handed analogue.  Closing an intermediate product at a boundary can change
the data seen by the next boundary product.

The output of the section is therefore a three-level type calculus, with the
middle level as its associative and residuated part.  Section~\ref{sec:derived}
uses the middle-profile coordinates introduced in Section~\ref{sec:middle}.  A
middle type \(A\), together with its peripheral complement, is represented by
the nuclear point
\[
  (\mu_A,\Mtc^*\mu_A)\in\Nuc(\Mtc).
\]
The second coordinate
\[
  \Mtc^*\mu_A\colon\cat C\times\cat C\to\Rbar
\]
is a binary kernel in the two boundary variables.  Viewed as a profunctor
\(\cat C\nrightarrow\cat C\), it has its own Isbell nucleus.
Section~\ref{sec:derived} studies these derived nuclei and rewrites the product
and residuals assembled above in those coordinates.

\begin{remark}
  Further operations can be considered between peripheral types, but their
  closure and variance hypotheses require a separate analysis.  The present
  section records the unit and residual structure needed for the middle calculus.
\end{remark}

\section{Derived nuclei in presheaf coordinates}\label{sec:derived}

The ternary measurement defines an \(\Rbar\)-profunctor
\[
  \Mtc\colon \cat C\nrightarrow\cat C\times\cat C,
  \qquad
  \Mtc\bigl(b,(x,z)\bigr)=M_3(x,b,z),
\]
where
\[
  M_3(x,b,z)=p(xbz)-p(x)-p(b)-p(z).
\]
We have two pictures of middle types.  In the type picture, a middle type is a
subset of \(\cat C\times\R\) closed under \(\Perp\)-biorthogonality;
Sections~\ref{sec:middle}--\ref{sec:categorical-structure} built the Lambek
calculus of middle types in that form.  In the Isbell picture, the same middle
type is represented, uusing the profile dictionary of Sections~\ref{sec:correspondence} and~\ref{sec:middle}, by a point
\[
  (f,g)\in\Nuc(\Mtc),
\]
where \(f\colon\cat C\to\Rbar\) is the presheaf coordinate and
\(g=\Mtc^*f\colon\cat C\times\cat C\to\Rbar\) is the copresheaf coordinate. The second coordinate \(g\) is itself a binary profunctor
\[
  g\colon\cat C\nrightarrow\cat C,
\]
which has its own Isbell nucleus.  We call \(\Nuc(g)\) the nucleus \emph{derived} from the middle nuclear point \((f,g)\).  If this point comes from a middle
type \(A\), so that \(f=\varphi_A\) and \(g=\Mtc^*f\), we also call
\(\Nuc(g)\) the nucleus derived from \(A\).  Its points will be called derived
types.

The first result of the section shows that the nucleus derived from a middle
point is computed by the two residuals \(f\multimap_l\) and
\(f\multimap_r\) of the Lambek calculus.
The ternary measurement \(M_3\) also
has two other one-variable arrangements, obtained by singling out the left or
the right coordinate rather than the middle one.  These give two further
iterated nuclei.  The three constructions are then compared by a
two-out-of-three theorem: compatibility with any two arrangements forces
compatibility with the third.  Their common intersection is the balanced locus
computed below.

We begin by recording the formulas that translate the type calculus into Isbell
coordinates.
It will be convenient to recall the trefoil identities:
\begin{equation}\label{eq:trefoil-identities}
  M_3(x,b,z)
  =
  M(x,bz)+M(b,z)
  =
  M(xb,z)+M(x,b).
\end{equation}
There are higher splitting identities as well.  For
\[
  M_4(a,b,c,d):=p(abcd)-p(a)-p(b)-p(c)-p(d),
\]
two adjacent variables may be contracted to their product:
\begin{equation}\label{eq:M4-adjacent-splitting}
  \begin{aligned}
    M_4(a,b,c,d)
     & =
    M_3(ab,c,d)+M(a,b) \\
     & =
    M_3(a,bc,d)+M(b,c) \\
     & =
    M_3(a,b,cd)+M(c,d).
  \end{aligned}
\end{equation}
Both \eqref{eq:trefoil-identities} and \eqref{eq:M4-adjacent-splitting}
follow immediately from the definitions of the measurements and associativity
of the execution product.

We write \(r\star s\) for the raw product profile
\(\pi_{r,s}\) of Section~\ref{sec:execution-types}:
\begin{equation}\label{eq:raw-convolution-star}
  (r\star s)(y)
  :=
  \pi_{r,s}(y)
  =
  \sup_{ab=y}
  \bigl(r(a)+s(b)-M(a,b)\bigr).
\end{equation}
This operation is a coordinate-level operation---to get a nucleus point
requires applying the relevant Isbell closure.

For example, let \(A\) and \(A'\) be middle types, and write
\[
  f=\varphi_A,
  \qquad
  f'=\varphi_{A'}
\]
for their presheaf coordinates.  Let
\[
  A''=A\odot A',
  \qquad
  f''=\varphi_{A''}.
\]
The product \(A''=A\odot A'\) is obtained by closing the raw coordinate
\(f\star f'\) in the middle nucleus:
\[
  f''
  =
  \Mtc_*\Mtc^*(f\star f').
\]

The copresheaf coordinate of this product has a single \(M_4\)-formula.  Define
\begin{equation}\label{eq:middle-product-boundary-kernel-def}
  P_{f,f'}(x,z)
  :=
  \inf_{a,a'\in\cat C}
  \bigl(
  M_4(x,a,a',z)-f(a)-f'(a')
  \bigr).
\end{equation}
Then
\begin{equation}\label{eq:middle-product-boundary-kernel}
  \Mtc^*(f\star f')=P_{f,f'}.
\end{equation}
Indeed,
\[
  \begin{aligned}
    \Mtc^*(f\star f')(x,z)
     & =
    \inf_y\bigl(M_3(x,y,z)-(f\star f')(y)\bigr) \\
     & =
    \inf_{a,a'}
    \bigl(
    M_3(x,aa',z)+M(a,a')-f(a)-f'(a')
    \bigr)                                      \\
     & =
    \inf_{a,a'}
    \bigl(
    M_4(x,a,a',z)-f(a)-f'(a')
    \bigr)                                      \\
     & =
    P_{f,f'}(x,z).
  \end{aligned}
\]
Here we use \(r-\sup_i s_i=\inf_i(r-s_i)\) and the middle
contraction in \eqref{eq:M4-adjacent-splitting}.  Consequently
\[
  f''=\Mtc_*P_{f,f'},
  \qquad
  A''=\Omega_{\Mtc_*P_{f,f'}}.
\]
Since \(P_{f,f'}=\Mtc^*(f\star f')\), the pair
\[
  \bigl(\Mtc_*P_{f,f'},P_{f,f'}\bigr)
\]
lies in \(\Nuc(\Mtc)\).  This is the middle nuclear point corresponding to
\(A\odot A'\).  Its copresheaf coordinate
\(\cat C\times\cat C\to\Rbar\) is the kernel \(P_{f,f'}\).

Finally, the residuals of the raw convolution are given by
\[
  (f\multimap_l r)(z)
  =
  \inf_{b\in\cat C}
  \bigl(r(bz)+M(b,z)-f(b)\bigr),
\]
and
\[
  (f\multimap_r \ell)(x)
  =
  \inf_{b\in\cat C}
  \bigl(\ell(xb)+M(x,b)-f(b)\bigr).
\]
They are characterized by the residuation laws
\[
  f\star v\le r
  \quad\Longleftrightarrow\quad
  v\le f\multimap_l r,
\]
and
\[
  u\star f\le \ell
  \quad\Longleftrightarrow\quad
  u\le f\multimap_r \ell.
\]
When the inputs are coordinates of the appropriate types, these coordinate
operations are the two residuals of the middle Lambek calculus.

\subsection{The middle derived nucleus}
\label{subsec:middle-derived-nucleus}

The first calculation concerns the derived nucleus of a middle nuclear point.

\begin{theorem}[Derived residual formulas]\label{thm:derived-formula}
  Let \((f,g)\in\Nuc(\Mtc)\), and let \((u,v)\in\Nuc(g)\).  Then
  \[
    v=f\multimap_l M^*u,
    \qquad
    u=f\multimap_r M_*v.
  \]
\end{theorem}

\begin{proof}
  Since \((f,g)\in\Nuc(\Mtc)\), we have
  \[
    g(x,z)
    =
    \inf_b
    \bigl(M_3(x,b,z)-f(b)\bigr).
  \]
  Since \((u,v)\in\Nuc(g)\), we have
  \[
    v(z)=\inf_x\bigl(g(x,z)-u(x)\bigr),
    \qquad
    u(x)=\inf_z\bigl(g(x,z)-v(z)\bigr).
  \]
  Therefore
  \[
    v(z)
    =
    \inf_{x,b}
    \bigl(M_3(x,b,z)-u(x)-f(b)\bigr).
  \]
  Using the trefoil identity
  \[
    M_3(x,b,z)=M(x,bz)+M(b,z),
  \]
  we obtain
  \[
    \begin{aligned}
      v(z)
       & =
      \inf_{x,b}
      \bigl(M(x,bz)+M(b,z)-u(x)-f(b)\bigr) \\
       & =
      \inf_b
      \bigl((M^*u)(bz)+M(b,z)-f(b)\bigr)   \\
       & =
      (f\multimap_l M^*u)(z).
    \end{aligned}
  \]
  Similarly,
  \[
    u(x)
    =
    \inf_{b,z}
    \bigl(M_3(x,b,z)-f(b)-v(z)\bigr).
  \]
  Using the other trefoil identity
  \[
    M_3(x,b,z)=M(xb,z)+M(x,b),
  \]
  gives
  \[
    \begin{aligned}
      u(x)
       & =
      \inf_{b,z}
      \bigl(M(xb,z)+M(x,b)-f(b)-v(z)\bigr) \\
       & =
      \inf_b
      \bigl((M_*v)(xb)+M(x,b)-f(b)\bigr)   \\
       & =
      (f\multimap_r M_*v)(x).
    \end{aligned}
  \]
\end{proof}

\begin{corollary}[Binary nuclearity of derived coordinates]
  \label{cor:nuclearity-derived}
  Let \((f,g)\in\Nuc(\Mtc)\), and let \((u,v)\in\Nuc(g)\).  Then
  \[
    u=M_*(f\star v),
    \qquad
    v=M^*(u\star f).
  \]
  Consequently
  \[
    M_*M^*u=u,
    \qquad
    M^*M_*v=v.
  \]
  Thus the two coordinates of a derived type are separately fixed by the
  binary Isbell closures:
  \[
    u\in\fix(M_*M^*),
    \qquad
    v\in\fix(M^*M_*).
  \]
  Equivalently,
  \[
    (u,M^*u)\in\Nuc(M),
    \qquad
    (M_*v,v)\in\Nuc(M).
  \]
  The derived pair \((u,v)\) need not itself lie in \(\Nuc(M)\).
\end{corollary}

\begin{proof}
  From the proof of Theorem~\ref{thm:derived-formula},
  \[
    u(x)
    =
    \inf_{b,z}
    \bigl(M_3(x,b,z)-f(b)-v(z)\bigr).
  \]
  Using
  \[
    M_3(x,b,z)=M(x,bz)+M(b,z),
  \]
  and regrouping by \(y=bz\), we obtain
  \[
    \begin{aligned}
      u(x)
       & =
      \inf_{b,z}
      \bigl(M(x,bz)+M(b,z)-f(b)-v(z)\bigr) \\
       & =
      \inf_y
      \left(
      M(x,y)
      -
      \sup_{bz=y}
      \bigl(f(b)+v(z)-M(b,z)\bigr)
      \right)                              \\
       & =
      M_*(f\star v)(x).
    \end{aligned}
  \]
  Similarly, from the proof of Theorem~\ref{thm:derived-formula},
  \[
    v(z)
    =
    \inf_{x,b}
    \bigl(M_3(x,b,z)-u(x)-f(b)\bigr).
  \]
  Using
  \[
    M_3(x,b,z)=M(xb,z)+M(x,b),
  \]
  and regrouping by \(y=xb\), we obtain
  \[
    \begin{aligned}
      v(z)
       & =
      \inf_{x,b}
      \bigl(M(xb,z)+M(x,b)-u(x)-f(b)\bigr) \\
       & =
      \inf_y
      \left(
      M(y,z)
      -
      \sup_{xb=y}
      \bigl(u(x)+f(b)-M(x,b)\bigr)
      \right)                              \\
       & =
      M^*(u\star f)(z).
    \end{aligned}
  \]
  Thus \(u\) lies in the image of \(M_*\), and \(v\) lies in the image of
  \(M^*\).  By the fixed-point characterization of the Isbell images,
  \[
    M_*M^*u=u,
    \qquad
    M^*M_*v=v.
  \]
\end{proof}

Thus the derived nucleus of a middle point is computed by the two residuals of
the middle Lambek calculus.  When the point $(u,v)$ is a type derived from a middle type \(A\), the types $\Omega_u$ and $\Omega_v$ are left and right types of the original realisability situation $(\cat C, \Ex, p)$.

\subsection{The two boundary arrangements}
\label{subsec:ternary-boundary-modules}
The ternary measurement $M_3$ has two other one-variable arrangements besides the middle-periphery arrangement $\Mtc$:
\[
  L(x,(b,z)):=M_3(x,b,z),
  \qquad
  R(z,(x,b)):=M_3(x,b,z).
\]

\begin{theorem}[Boundary residual formulas]
  \label{thm:boundary-residual-formulas}
  The following formulas hold.

  \begin{enumerate}
    \item[\textup{(L)}]
          Let \((u,h)\in\Nuc(L)\), and let \((f,v)\in\Nuc(h)\).  Then
          \[
            v=f\multimap_l M^*u.
          \]
          Equivalently,
          \[
            v(z)
            =
            \inf_b
            \bigl((M^*u)(bz)+M(b,z)-f(b)\bigr).
          \]
          The middle coordinate also satisfies the two boundary formulas
          \[
            f(b)
            =
            \inf_x
            \bigl((M_*v)(xb)+M(x,b)-u(x)\bigr)
            =
            \inf_z
            \bigl((M^*u)(bz)+M(b,z)-v(z)\bigr).
          \]

    \item[\textup{(R)}]
          Let \((v,k)\in\Nuc(R)\), and let \((u,f)\in\Nuc(k)\).  Then
          \[
            u=f\multimap_r M_*v.
          \]
          Equivalently,
          \[
            u(x)
            =
            \inf_b
            \bigl((M_*v)(xb)+M(x,b)-f(b)\bigr).
          \]
          The middle coordinate satisfies the same two boundary formulas
          \[
            f(b)
            =
            \inf_x
            \bigl((M_*v)(xb)+M(x,b)-u(x)\bigr)
            =
            \inf_z
            \bigl((M^*u)(bz)+M(b,z)-v(z)\bigr).
          \]
  \end{enumerate}
\end{theorem}

\begin{proof}
  For the left arrangement, since \((u,h)\in\Nuc(L)\),
  \[
    h(b,z)
    =
    \inf_x
    \bigl(M_3(x,b,z)-u(x)\bigr).
  \]
  Since \((f,v)\in\Nuc(h)\),
  \[
    f(b)=\inf_z\bigl(h(b,z)-v(z)\bigr),
    \qquad
    v(z)=\inf_b\bigl(h(b,z)-f(b)\bigr).
  \]
  Substitution gives
  \[
    f(b)
    =
    \inf_{x,z}
    \bigl(M_3(x,b,z)-u(x)-v(z)\bigr),
  \]
  and
  \[
    v(z)
    =
    \inf_{x,b}
    \bigl(M_3(x,b,z)-u(x)-f(b)\bigr).
  \]
  Using
  \[
    M_3(x,b,z)=M(x,bz)+M(b,z),
  \]
  in the formula for \(v\), we obtain
  \[
    \begin{aligned}
      v(z)
       & =
      \inf_{x,b}
      \bigl(M(x,bz)+M(b,z)-u(x)-f(b)\bigr) \\
       & =
      \inf_b
      \bigl((M^*u)(bz)+M(b,z)-f(b)\bigr)   \\
       & =
      (f\multimap_l M^*u)(z).
    \end{aligned}
  \]
  The two formulas for \(f\) follow from the two trefoil identities.  Using
  \[
    M_3(x,b,z)=M(xb,z)+M(x,b),
  \]
  gives
  \[
    f(b)
    =
    \inf_x
    \bigl((M_*v)(xb)+M(x,b)-u(x)\bigr),
  \]
  while using
  \[
    M_3(x,b,z)=M(x,bz)+M(b,z)
  \]
  gives
  \[
    f(b)
    =
    \inf_z
    \bigl((M^*u)(bz)+M(b,z)-v(z)\bigr).
  \]

  The right arrangement is analogous.  Since \((v,k)\in\Nuc(R)\),
  \[
    k(x,b)
    =
    \inf_z
    \bigl(M_3(x,b,z)-v(z)\bigr).
  \]
  Since \((u,f)\in\Nuc(k)\),
  \[
    u(x)=\inf_b\bigl(k(x,b)-f(b)\bigr),
    \qquad
    f(b)=\inf_x\bigl(k(x,b)-u(x)\bigr).
  \]
  Hence
  \[
    u(x)
    =
    \inf_{b,z}
    \bigl(M_3(x,b,z)-f(b)-v(z)\bigr),
  \]
  and
  \[
    f(b)
    =
    \inf_{x,z}
    \bigl(M_3(x,b,z)-u(x)-v(z)\bigr).
  \]
  Using
  \[
    M_3(x,b,z)=M(xb,z)+M(x,b)
  \]
  in the formula for \(u\), we get
  \[
    \begin{aligned}
      u(x)
       & =
      \inf_{b,z}
      \bigl(M(xb,z)+M(x,b)-f(b)-v(z)\bigr) \\
       & =
      \inf_b
      \bigl((M_*v)(xb)+M(x,b)-f(b)\bigr)   \\
       & =
      (f\multimap_r M_*v)(x).
    \end{aligned}
  \]
  The two formulas for \(f\) are obtained from the same two trefoil identities,
  exactly as in the left arrangement.
\end{proof}

\subsection{Two-out-of-three coherence}
\label{subsec:triple-coincidence}

Thus the three one-variable arrangements of \(M_3\) give three ways of
producing triples \((u,f,v)\).  The middle arrangement starts from a middle
coordinate \(f\) and produces the two boundary coordinates.  The left boundary
arrangement starts from \(u\) and produces \(f\) and \(v\).  The right boundary
arrangement starts from \(v\) and produces \(u\) and \(f\).  We now compare the
resulting triples. Let
\[
  \begin{aligned}
    T_L
     & :=
    \left\{(u,f,v)\;\middle|\;
    (u,L^*u)\in\Nuc(L),\ (f,v)\in\Nuc(L^*u)
    \right\}, \\
    T_{\Mtc}
     & :=
    \left\{(u,f,v)\;\middle|\;
    (f,\Mtc^*f)\in\Nuc(\Mtc),\ (u,v)\in\Nuc(\Mtc^*f)
    \right\}, \\
    T_R
     & :=
    \left\{(u,f,v)\;\middle|\;
    (v,R^*v)\in\Nuc(R),\ (u,f)\in\Nuc(R^*v)
    \right\}.
  \end{aligned}
\]
These are the triples obtained by iterating the nuclei of the left boundary,
middle, and right boundary arrangements. For a triple \((u,f,v)\), consider the following three residual equations:
\begin{align}
  u(x)
   & =
  \inf_{b,z}
  \bigl(M_3(x,b,z)-f(b)-v(z)\bigr),
  \tag{\(B_u\)}
  \label{eq:balanced-u}
  \\[1mm]
  f(b)
   & =
  \inf_{x,z}
  \bigl(M_3(x,b,z)-u(x)-v(z)\bigr),
  \tag{\(B_f\)}
  \label{eq:balanced-f}
  \\[1mm]
  v(z)
   & =
  \inf_{x,b}
  \bigl(M_3(x,b,z)-u(x)-f(b)\bigr).
  \tag{\(B_v\)}
  \label{eq:balanced-v}
\end{align}
We call \((u,f,v)\) a balanced triple for \(M_3\) if it satisfies all three
equations.  Let \(T_{\mathrm{bal}}\) denote the set of balanced triples.

\begin{lemma}[Equations seen by the three arrangements]
  \label{lem:arrangements-see-balanced-equations}
  The three arrangements satisfy the following implications:
  \[
    \begin{aligned}
      (u,f,v)\in T_L
       & \implies
      \eqref{eq:balanced-f}\text{ and }\eqref{eq:balanced-v}, \\
      (u,f,v)\in T_{\Mtc}
       & \implies
      \eqref{eq:balanced-u}\text{ and }\eqref{eq:balanced-v}, \\
      (u,f,v)\in T_R
       & \implies
      \eqref{eq:balanced-u}\text{ and }\eqref{eq:balanced-f}.
    \end{aligned}
  \]
\end{lemma}

\begin{proof}
  Suppose first that \((u,f,v)\in T_L\).  Put \(h=L^*u\).  Then
  \[
    h(b,z)
    =
    \inf_x
    \bigl(M_3(x,b,z)-u(x)\bigr),
  \]
  and \((f,v)\in\Nuc(h)\).  Hence
  \[
    \begin{aligned}
      f(b)
       & =
      \inf_z\bigl(h(b,z)-v(z)\bigr) \\
       & =
      \inf_{x,z}
      \bigl(M_3(x,b,z)-u(x)-v(z)\bigr),
    \end{aligned}
  \]
  and
  \[
    \begin{aligned}
      v(z)
       & =
      \inf_b\bigl(h(b,z)-f(b)\bigr) \\
       & =
      \inf_{x,b}
      \bigl(M_3(x,b,z)-u(x)-f(b)\bigr).
    \end{aligned}
  \]
  Thus \eqref{eq:balanced-f} and \eqref{eq:balanced-v} hold. The middle arrangement is the same calculation with
  \(g=\Mtc^*f\).  If \((u,f,v)\in T_{\Mtc}\), then
  \[
    g(x,z)
    =
    \inf_b
    \bigl(M_3(x,b,z)-f(b)\bigr),
  \]
  and \((u,v)\in\Nuc(g)\).  Therefore
  \[
    u(x)
    =
    \inf_{b,z}
    \bigl(M_3(x,b,z)-f(b)-v(z)\bigr),
  \]
  and
  \[
    v(z)
    =
    \inf_{x,b}
    \bigl(M_3(x,b,z)-u(x)-f(b)\bigr).
  \]
  Thus \eqref{eq:balanced-u} and \eqref{eq:balanced-v} hold. Finally, suppose that \((u,f,v)\in T_R\).  Put \(k=R^*v\).  Then
  \[
    k(x,b)
    =
    \inf_z
    \bigl(M_3(x,b,z)-v(z)\bigr),
  \]
  and \((u,f)\in\Nuc(k)\).  Hence
  \[
    u(x)
    =
    \inf_{b,z}
    \bigl(M_3(x,b,z)-f(b)-v(z)\bigr),
  \]
  and
  \[
    f(b)
    =
    \inf_{x,z}
    \bigl(M_3(x,b,z)-u(x)-v(z)\bigr).
  \]
  Thus \eqref{eq:balanced-u} and \eqref{eq:balanced-f} hold.
\end{proof}

\begin{theorem}[Two-out-of-three coherence]
  \label{thm:triple-coincidence}
  The three iterated ternary constructions have the same pairwise intersection:
  \[
    T_L\cap T_{\Mtc}
    =
    T_{\Mtc}\cap T_R
    =
    T_L\cap T_R
    =
    T_L\cap T_{\Mtc}\cap T_R
    =
    T_{\mathrm{bal}}.
  \]
\end{theorem}

\begin{proof}
  By Lemma~\ref{lem:arrangements-see-balanced-equations}, membership in any two
  of \(T_L,T_{\Mtc},T_R\) gives all three equations
  \eqref{eq:balanced-u}, \eqref{eq:balanced-f}, and
  \eqref{eq:balanced-v}.  Hence every pairwise intersection is contained in
  \(T_{\mathrm{bal}}\). Conversely, suppose that \((u,f,v)\in T_{\mathrm{bal}}\).  We show that it
  lies in all three arrangement sets. First put
  \[
    g=\Mtc^*f.
  \]
  Then
  \[
    g(x,z)
    =
    \inf_b
    \bigl(M_3(x,b,z)-f(b)\bigr).
  \]
  Equations \eqref{eq:balanced-u} and \eqref{eq:balanced-v} say exactly that
  \[
    u(x)=\inf_z\bigl(g(x,z)-v(z)\bigr),
    \qquad
    v(z)=\inf_x\bigl(g(x,z)-u(x)\bigr).
  \]
  Thus \((u,v)\in\Nuc(g)\), once we know that \((f,g)\in\Nuc(\Mtc)\). Since \(g=\Mtc^*f\), the closure inequality gives
  \[
    f\le \Mtc_*g.
  \]
  Equation \eqref{eq:balanced-v} implies that, for all \(x,b,z\),
  \[
    u(x)+v(z)\le M_3(x,b,z)-f(b).
  \]
  Taking the infimum over \(b\) gives
  \[
    u(x)+v(z)\le g(x,z).
  \]
  Therefore
  \[
    \begin{aligned}
      \Mtc_*g(b)
       & =
      \inf_{x,z}
      \bigl(M_3(x,b,z)-g(x,z)\bigr)    \\
       & \le
      \inf_{x,z}
      \bigl(M_3(x,b,z)-u(x)-v(z)\bigr) \\
       & =
      f(b),
    \end{aligned}
  \]
  where the last equality is \eqref{eq:balanced-f}.  Hence
  \[
    f=\Mtc_*g.
  \]
  Thus \((f,g)\in\Nuc(\Mtc)\), and consequently
  \[
    (u,f,v)\in T_{\Mtc}.
  \]
  Next put
  \[
    h=L^*u.
  \]
  Then
  \[
    h(b,z)
    =
    \inf_x
    \bigl(M_3(x,b,z)-u(x)\bigr).
  \]
  Equations \eqref{eq:balanced-f} and \eqref{eq:balanced-v} say exactly that
  \[
    f(b)=\inf_z\bigl(h(b,z)-v(z)\bigr),
    \qquad
    v(z)=\inf_b\bigl(h(b,z)-f(b)\bigr).
  \]
  Thus \((f,v)\in\Nuc(h)\), once we know that \((u,h)\in\Nuc(L)\). Since \(h=L^*u\), the closure inequality gives
  \[
    u\le L_*h.
  \]
  Equation \eqref{eq:balanced-f} implies that, for all \(x,b,z\),
  \[
    f(b)+v(z)\le M_3(x,b,z)-u(x).
  \]
  Taking the infimum over \(x\) gives
  \[
    f(b)+v(z)\le h(b,z).
  \]
  Therefore
  \[
    \begin{aligned}
      L_*h(x)
       & =
      \inf_{b,z}
      \bigl(M_3(x,b,z)-h(b,z)\bigr)    \\
       & \le
      \inf_{b,z}
      \bigl(M_3(x,b,z)-f(b)-v(z)\bigr) \\
       & =
      u(x),
    \end{aligned}
  \]
  where the last equality is \eqref{eq:balanced-u}.  Hence
  \[
    u=L_*h.
  \]
  Thus \((u,h)\in\Nuc(L)\), and consequently
  \[
    (u,f,v)\in T_L.
  \]
  Finally put
  \[
    k=R^*v.
  \]
  Then
  \[
    k(x,b)
    =
    \inf_z
    \bigl(M_3(x,b,z)-v(z)\bigr).
  \]
  Equations \eqref{eq:balanced-u} and \eqref{eq:balanced-f} say exactly that
  \[
    u(x)=\inf_b\bigl(k(x,b)-f(b)\bigr),
    \qquad
    f(b)=\inf_x\bigl(k(x,b)-u(x)\bigr).
  \]
  Thus \((u,f)\in\Nuc(k)\), once we know that \((v,k)\in\Nuc(R)\). Since \(k=R^*v\), the closure inequality gives
  \[
    v\le R_*k.
  \]
  Equation \eqref{eq:balanced-u} implies that, for all \(x,b,z\),
  \[
    u(x)+f(b)\le M_3(x,b,z)-v(z).
  \]
  Taking the infimum over \(z\) gives
  \[
    u(x)+f(b)\le k(x,b).
  \]
  Therefore
  \[
    \begin{aligned}
      R_*k(z)
       & =
      \inf_{x,b}
      \bigl(M_3(x,b,z)-k(x,b)\bigr)    \\
       & \le
      \inf_{x,b}
      \bigl(M_3(x,b,z)-u(x)-f(b)\bigr) \\
       & =
      v(z),
    \end{aligned}
  \]
  where the last equality is \eqref{eq:balanced-v}.  Hence
  \[
    v=R_*k.
  \]
  Thus \((v,k)\in\Nuc(R)\), and consequently
  \[
    (u,f,v)\in T_R.
  \]

  Therefore every balanced triple lies in
  \(T_L\cap T_{\Mtc}\cap T_R\).  The claimed equalities follow.
\end{proof}

The common intersection is therefore explicit: it is the locus where the three
one-coordinate residual equations of \(M_3\) are simultaneously saturated.
Equivalently, on \(T_{\mathrm{bal}}\), each coordinate is recovered from the
other two:
\[
  u(x)
  =
  \inf_{b,z}
  \bigl(M_3(x,b,z)-f(b)-v(z)\bigr),
\]
\[
  f(b)
  =
  \inf_{x,z}
  \bigl(M_3(x,b,z)-u(x)-v(z)\bigr),
\]
and
\[
  v(z)
  =
  \inf_{x,b}
  \bigl(M_3(x,b,z)-u(x)-f(b)\bigr).
\]
Thus any one coordinate may be taken as the coordinate from which the other two
are derived: \((f,v)\) is derived from \(u\) in the left arrangement,
\((u,v)\) is derived from \(f\) in the middle arrangement, and \((u,f)\) is
derived from \(v\) in the right arrangement.  Moreover, within the balanced
locus, the derived pair uniquely recovers the coordinate from which it is
derived.

The next example shows why the last qualification is necessary.  A pair of
coordinates can be derived from a middle coordinate \(f\), while the middle
coordinate recovered from that pair is a different coordinate \(f'\).

\subsection{Sharpness of the two-out-of-three hypothesis}
\label{subsec:balanced-sharpness}

The following example shows that the two-out-of-three hypothesis in
Theorem~\ref{thm:triple-coincidence} cannot be weakened to membership in a
single arrangement.  It also shows that a pair derived from a middle nucleus
need not itself be a binary nuclear point for the original kernel \(M\).
We use the same four-element monoid as in
Example~\ref{ex:left-product-nonassociativity}. Let \(\cat C=\{e,a,c,d\}\), with multiplication table
\[
  \begin{array}{c|cccc}
    \cdot & e & a & c & d \\
    \hline
    e     & e & a & c & d \\
    a     & a & e & c & d \\
    c     & c & d & c & d \\
    d     & d & c & c & d
  \end{array}
\]
and define \(p:\cat C \to \R\) by
\[
  p(e)=0,\qquad p(a)=-3,\qquad p(c)=1,\qquad p(d)=4.
\]
The table is associative, and \(e\) is the unit. Write vectors in the order \((e,a,c,d)\), and set
\[
  u=(0,0,-1,-4),
  \qquad
  f=(-6,0,-7,-10),
  \qquad
  v=(0,6,3,3).
\]
The middle coordinate \(f\) is the middle closure of the principal seed at
\(a\). Put
\[
  g=\Mtc^*f.
\]
Direct calculation gives
\[
  \begin{array}{c|rrrr}
    g & e & a & c  & d  \\
    \hline
    e & 0 & 6 & 3  & 3  \\
    a & 6 & 6 & 6  & 6  \\
    c & 6 & 6 & 2  & 2  \\
    d & 0 & 6 & -1 & -1
  \end{array}.
\]
Direct calculation gives
\[
  \Mtc_*g=f,
  \qquad
  g^*u=v,
  \qquad
  g_*v=u.
\]
Hence
\[
  (f,g)\in\Nuc(\Mtc),
  \qquad
  (u,v)\in\Nuc(g),
\]
and therefore
\[
  (u,f,v)\in T_{\Mtc}.
\]
However, the missing middle-coordinate equation gives
\[
  \left(
  \inf_{x,z}
  \bigl(M_3(x,b,z)-u(x)-v(z)\bigr)
  \right)_b
  =
  (-6,0,-4,-7),
\]
which is not
\[
  f=(-6,0,-7,-10).
\]
The displayed vector is pointwise larger than \(f\), strictly at \(c\) and
\(d\).  Thus \eqref{eq:balanced-f} fails, and hence
\[
  (u,f,v)\notin T_{\mathrm{bal}}.
\]
Consequently
\[
  (u,f,v)\in T_{\Mtc}\setminus T_{\mathrm{bal}}.
\]
The same example also separates derived nuclearity from binary nuclearity for
the original kernel \(M\).  Although
\[
  (u,v)\in\Nuc(g),
\]
we have
\[
  M^*u=(0,0,0,0)\neq (0,6,3,3)=v,
\]
and
\[
  M_*v=(-6,0,-4,-7)\neq (0,0,-1,-4)=u.
\]
Therefore
\[
  (u,v)\notin\Nuc(M).
\]

\section{Product envelopes and factor gaps}\label{sec:product-geometry}

In the finite real case the presheaf-coordinate calculus admits a projective
bookkeeping of product witnesses.  In this section, we give the finite coordinate
definitions needed for the product-envelope and factor-gap formulas.  The gap
quantities below are the invariants that connect the middle product with the
polyhedral and metric geometry of projective nuclei developed in the
paper \cite{GastaldiJarvisSeillerTerillaNucleus}.

Throughout this section assume that \(\cat C\) is finite and that
\(p\colon\cat C\to\R\) is real-valued.  Thus the ternary measurement \(\Mtc\)
is a finite real matrix.  Let \(\Nuc_{\R}(\Mtc)\) be the locus of points
\((f,g)\in\Nuc(\Mtc)\) whose presheaf and copresheaf coordinates are finite
real-valued functions.  Define the projective nucleus used in this section by
\[
  \pnuc(\Mtc):=\Nuc_{\R}(\Mtc)/{\sim},
\]
where
\[
  (f,g)\sim(f+\lambda,g-\lambda),
  \qquad \lambda\in\R.
\]
For \(X=(f,g)\in\Nuc_{\R}(\Mtc)\), write \([X]\) for its projective class. We keep the standard nuclear coordinates.  Let
\[
  X=(f,g),\qquad Y=(f',g')
\]
be points of \(\Nuc_{\R}(\Mtc)\), so
\[
  g=\Mtc^*f,
  \qquad
  g'=\Mtc^*f'.
\]
Define
\[
  M_4(a,b,c,d):=p(abcd)-p(a)-p(b)-p(c)-p(d).
\]
The product
\[
  Z:=X\odot Y=(f'',g'')
\]
is represented by the nuclear pair whose copresheaf coordinate is
\begin{equation}\label{eq:product-envelope}
  g''(x,z)
  =
  \min_{b,b'\in\cat C}
  \left(
  M_4(x,b,b',z)-f(b)-f'(b')
  \right),
\end{equation}
and whose presheaf coordinate is recovered by
\begin{equation}\label{eq:product-presheaf-from-envelope}
  f''(y)=\Mtc_*g''(y)
  =
  \min_{x,z\in\cat C}
  \left(M_3(x,y,z)-g''(x,z)\right).
\end{equation}
Equivalently, \(g''\) is the boundary envelope obtained by minimizing over the
two middle factors, and \(f''\) is the middle closure recovered from that
envelope.

\begin{corollary}[Projective product envelope]
  \label{cor:projective-middle-product}
  The assignment
  \[
    \mu_\odot\colon
    \pnuc(\Mtc)\times\pnuc(\Mtc)
    \longrightarrow
    \pnuc(\Mtc),
    \qquad
    ([X],[Y])\longmapsto [Z]
  \]
  is well-defined.  More explicitly, if
  \[
    X=(f,g)
    \quad\text{is replaced by}\quad
    X_\lambda=(f+\lambda,g-\lambda)
  \]
  and
  \[
    Y=(f',g')
    \quad\text{is replaced by}\quad
    Y_\mu=(f'+\mu,g'-\mu),
  \]
  then the output coordinates transform as
  \[
    g''\mapsto g''-\lambda-\mu,
    \qquad
    f''\mapsto f''+\lambda+\mu.
  \]
  Consequently the product gap
  \begin{equation}\label{eq:product-gap}
    \Delta_{X,Y}(y;x,z)
    :=
    M_3(x,y,z)-f''(y)-g''(x,z)
  \end{equation}
  is independent of the chosen affine representatives of \([X]\) and \([Y]\).
\end{corollary}

\begin{proof}
  Formula \eqref{eq:product-envelope} is the middle-product convolution formula
  in the finite real-valued case.  Replacing \(f\) by \(f+\lambda\) and \(f'\)
  by \(f'+\mu\) subtracts \(\lambda+\mu\) from every term in the minimum
  defining \(g''\), hence sends \(g''\) to \(g''-\lambda-\mu\).  Applying
  \(\Mtc_*\) then sends \(f''\) to \(f''+\lambda+\mu\).  Thus the projective
  class \([Z]\) is independent of the representatives.  The same two shifts
  cancel in \(M_3(x,y,z)-f''(y)-g''(x,z)\), proving the invariance of
  \(\Delta_{X,Y}\).
\end{proof}

\subsection{Factor gaps}
\label{subsec:product-gap-calculus}

The gap \(\Delta_{X,Y}\) is the ordinary middle gap of the output point
\(Z=(f'',g'')\).  To see how this output gap is built from the two input points,
we also keep track of the factorisations that compute the envelope \(g''\). For \(b,b',x,z\in\cat C\), define the factor gap
\begin{equation}\label{eq:convolution-gap}
  \kappa_{X,Y}^{b,b'}(x,z)
  :=
  M_4(x,b,b',z)-f(b)-f'(b')-g''(x,z).
\end{equation}
Then \(\kappa_{X,Y}^{b,b'}(x,z)\ge 0\), and it vanishes exactly when
\((b,b')\) realizes the minimum in \eqref{eq:product-envelope} at \((x,z)\). The factor gap decomposes into input witness gaps and outer minimization slacks.
Write
\begin{align}
  \delta_X(b;x,z)
   & :={}
  M_3(x,b,z)-f(b)-g(x,z),
  \label{eq:input-gap-X} \\
  \delta_Y(b;x,z)
   & :={}
  M_3(x,b,z)-f'(b)-g'(x,z).
  \label{eq:input-gap-Y}
\end{align}
These are the ordinary middle gap matrices of the input points \(X=(f,g)\) and
\(Y=(f',g')\).  Define also
\begin{align}
  \sigma^L_{X,Y}(b';x,z)
   & :={}
  g(x,b'z)+M(b',z)-f'(b')-g''(x,z),
  \label{eq:left-outer-slack} \\
  \sigma^R_{X,Y}(b;x,z)
   & :={}
  M(x,b)-f(b)+g'(xb,z)-g''(x,z).
  \label{eq:right-outer-slack}
\end{align}
The nonnegativity of these two slacks follows from the equivalent one-sided
forms of the envelope:
\begin{align}
  g''(x,z)
   & =
  \min_{b'\in\cat C}
  \left(g(x,b'z)+M(b',z)-f'(b')\right),
  \label{eq:left-envelope-form} \\
  g''(x,z)
   & =
  \min_{b\in\cat C}
  \left(M(x,b)-f(b)+g'(xb,z)\right).
  \label{eq:right-envelope-form}
\end{align}

\begin{proposition}[Factorisation of convolution gaps]
  \label{prop:convolution-gap-factorization}
  For all \(x,z,b,b'\in\cat C\),
  \begin{equation}\label{eq:convolution-gap-factorization}
    \kappa_{X,Y}^{b,b'}(x,z)
    =
    \delta_X(b;x,b'z)+\sigma^L_{X,Y}(b';x,z)
    =
    \delta_Y(b';xb,z)+\sigma^R_{X,Y}(b;x,z).
  \end{equation}
  Consequently, for fixed \(x,z,b,b'\), the following conditions are equivalent:
  \begin{enumerate}[label=\textup{(\roman*)}]
    \item \((b,b')\) computes the envelope \(g''(x,z)\), equivalently
          \(\kappa_{X,Y}^{b,b'}(x,z)=0\);
    \item \(\delta_X(b;x,b'z)=0\) and
          \(\sigma^L_{X,Y}(b';x,z)=0\);
    \item \(\delta_Y(b';xb,z)=0\) and
          \(\sigma^R_{X,Y}(b;x,z)=0\).
  \end{enumerate}
  In particular, every envelope-computing factorisation satisfies
  \[
    \delta_X(b;x,b'z)=
    \delta_Y(b';xb,z)=
    \sigma^L_{X,Y}(b';x,z)=
    \sigma^R_{X,Y}(b;x,z)=0.
  \]
\end{proposition}

\begin{proof}
  The adjacent splitting identities for \(M_4\) give
  \[
    M_4(x,b,b',z)=M_3(x,b,b'z)+M(b',z)
  \]
  and
  \[
    M_4(x,b,b',z)=M(x,b)+M_3(xb,b',z).
  \]
  Using the first identity,
  \begin{align*}
    \kappa_{X,Y}^{b,b'}(x,z)
     & =
    M_3(x,b,b'z)+M(b',z)-f(b)-f'(b')-g''(x,z)     \\
     & =
    \bigl(M_3(x,b,b'z)-f(b)-g(x,b'z)\bigr)        \\
     & \qquad
    +\bigl(g(x,b'z)+M(b',z)-f'(b')-g''(x,z)\bigr) \\
     & =
    \delta_X(b;x,b'z)+\sigma^L_{X,Y}(b';x,z).
  \end{align*}
  The second equality in \eqref{eq:convolution-gap-factorization} is obtained in
  the same way from the second adjacent splitting identity.  All summands are
  nonnegative: \(\delta_X\) and \(\delta_Y\) are ordinary middle gaps, while
  \(\sigma^L\) and \(\sigma^R\) are nonnegative by
  \eqref{eq:left-envelope-form} and \eqref{eq:right-envelope-form}.  Therefore
  \(\kappa_{X,Y}^{b,b'}(x,z)=0\) holds exactly when the two summands in either
  displayed decomposition vanish.
\end{proof}

Thus a product witness is not just a minimizing pair \((b,b')\).  It is a pair
for which \(b\) is visible in the witness geometry of \(X=(f,g)\) from the
shifted boundary \((x,b'z)\), \(b'\) is visible in the witness geometry of
\(Y=(f',g')\) from the shifted boundary \((xb,z)\), and the two outer envelope
minimizations are sharp.  The product gap \(\Delta_{X,Y}\) then records the
ordinary middle gap of the resulting output point \(Z=(f'',g'')\).

These identities are the promised local bridge between the Lambek product and
the geometry of execution.  They identify the witness data organized by the
chamber theory in the companion paper: the product witness, the two shifted input
witnesses, the outer envelope sharpness conditions, and the output gap.

\bibliographystyle{amsalpha}
\bibliography{references}

\end{document}